\documentclass[11pt]{article}
\usepackage[margin=1in]{geometry}

\usepackage[utf8]{inputenc}
\usepackage{amsmath,amssymb,amsfonts,amsthm,bm}
\usepackage{latexsym,bbm,graphicx,float}
\usepackage{mathtools,braket}
\usepackage[multiple]{footmisc} 
\usepackage{enumitem}
\usepackage{multirow, multicol}
\usepackage{thmtools, thm-restate}
\usepackage[plain]{algorithm}
\usepackage{algpseudocode}
\usepackage[usenames,dvipsnames,svgnames,table]{xcolor}
\usepackage{subcaption}
\usepackage{soul}
\usepackage{varwidth}
\usepackage{url}
\usepackage{soul} 
\usepackage{todonotes}
\usepackage{titlesec} 

\usepackage[backend=biber,backref=true,backrefstyle=two,style=alphabetic,maxalphanames=4,maxcitenames=2,maxbibnames=99,natbib=true]{biblatex}
\usepackage[breaklinks,hyperfootnotes=false]{hyperref}
\hypersetup{
	colorlinks   = true, 
	urlcolor     = Maroon, 
	linkcolor    = Maroon, 
	citecolor   =  Maroon 
}
\usepackage{cleveref}

\newtheorem{theorem}{Theorem}
\newtheorem{lemma}[theorem]{Lemma}

\theoremstyle{definition}\newtheorem{example}{Example}

\newtheorem{claim}{Claim}

\newtheorem{definition}{Definition}

\DeclareCaptionSubType*{table}

\algnewcommand\algprocedure{\textbf{Procedure:}}
\algnewcommand\Procedurename{\item[\underline{\algprocedure}]}
\algnewcommand\algmain{\textbf{Main:}}
\algnewcommand\Main{\item[\underline{\algmain}]}
\algnewcommand\algorithmicinput{\textbf{Input:}}
\algnewcommand\Input{\item[\algorithmicinput]}
\algnewcommand\algorithmicoutput{\textbf{Output:}}
\algnewcommand\Output{\item[\algorithmicoutput]}

\makeatletter
\renewenvironment{proof}[1][\proofname] {\pushQED{\qed}\normalfont\topsep\z@\@plus0\p@\relax\trivlist\item[\hskip\labelsep\bfseries#1\@addpunct{:}]\ignorespaces}{\popQED\endtrivlist\@endpefalse}
\makeatother

\makeatletter
\def\thm@space@setup{%
	\thm@preskip=8pt plus 2pt minus 4pt
	\thm@postskip=\thm@preskip 
}
\makeatother
\makeatletter
\renewcommand{\ALG@beginalgorithmic}{\small}
\makeatother

\titlespacing\section{0pt}{12pt plus 2pt minus 2pt}{-1pt plus1pt minus 1pt}
\titlespacing\subsection{0pt}{12pt plus 2pt minus 2pt}{-1pt plus 1pt minus 1pt}
\titlespacing\subsubsection{0pt}{12pt plus 2pt minus 2pt}{-1pt plus 1pt minus 1pt}

\setlist{topsep = 0pt plus1pt}

\newcommand{\abs}[1]{\left|#1\right|}
\newcommand{\floor}[1]{\left\lfloor#1\right\rfloor}
\newcommand{\ceil}[1]{\left\lceil#1\right\rceil}

\newcommand{\ma}{{\mathcal A}}

\newcommand{\me}{{\mathcal E}}
\newcommand{\mh}{{\mathcal H}}
\newcommand{\mi}{{\mathcal I}}

\newcommand{\ml}{{\mathcal L}}
\newcommand{\mm}{{\mathcal M}}
\newcommand{\mn}{{\mathcal N}}

\newcommand{\mpp}{{\mathcal P}}
\newcommand{\mq}{{\mathcal Q}}

\newcommand{\bbg}{{\mathbb G}}
\newcommand{\bmp}{\bm{p}}

\newcommand{\ttc}{\text{TTC}}
\newcommand{\sw}{\text{SW}}

\newcommand{\val}{\text{value}}

\newcommand{\proofcase}[1]{\par\smallskip\noindent{\textbf{#1}}}
\newcommand{\blank}[1]{}

\newenvironment{proof*}[1][\proofname]{
	
	\begin{proof}[#1]}{\end{proof}}

\allowdisplaybreaks

\urldef{\email}\path|{vijay.menon, kate.larson}@uwaterloo.ca|

\title{Improving Welfare in One-sided Matching using Simple Threshold Queries}

 \author{
 	Thomas Ma\footnote{Department of Computer Science, University of Toronto.\ \ttfamily thomas.ma@mail.utoronto.ca
 	}
 	\and
 	Vijay Menon\footnote{David R.\ Cheriton School of Computer Science, University of Waterloo.\ \email}
 	\and 
 	Kate Larson\footnotemark[2]
}

\date{}

\begin{document}
	
	\maketitle
	
	\begin{abstract}
	We study one-sided matching problems where $n$ agents have preferences over $m$ objects and each of them need to be assigned to at most one object. Most work on such problems assume that the agents only have ordinal preferences and usually the goal in them is to compute a matching that satisfies some notion of economic efficiency. However, in reality, agents may have some preference intensities or cardinal utilities that, e.g., indicate that they like an object much more than another object, and not taking these into account can result in a loss in welfare. While one way to potentially account for these is to directly ask the agents for this information, such an elicitation process is cognitively demanding. Therefore, we focus on learning more about their cardinal preferences using simple threshold queries which ask an agent if they value an object greater than a certain value, and use this in turn to come up with algorithms that produce a matching that, for a particular economic notion $X$, satisfies $X$ and also achieves a good approximation to the optimal welfare among all matchings that satisfy $X$. We focus on several notions of economic efficiency, and look at both adaptive and non-adaptive algorithms. Overall, our results show how one can improve welfare by even non-adaptively asking the agents for just one bit of extra information per object. 
\end{abstract}
	\section{Introduction}

There are several scenarios where there is a set of $n$ agents who have preferences over a set of $m$ objects and the goal is to assign each agent to at most one object. For instance, the case where the objects are houses represents the well-studied housing allocation problem (e.g., see \cite{hyll79,abd98,abra04}) or, when agents have initial endowments, the housing market problem (e.g., see \cite{shapley74,roth77,abd99}). Other examples include assigning faculty members to school committees, workers to tasks, etc. Together, these scenarios are commonly referred to as \textit{one-sided matching problems}. 

Most of the literature on one-sided matching problems typically assume that the agents have an acceptable set of objects and that they submit a preference order over the acceptable set. Given this, the standard objective is to come up with an assignment of agent to objects (henceforth, matching) such that it satisfies some notion of economic efficiency like, e.g., Pareto optimality or rank-maximality (formally defined in \Cref{sec:prelims}). Although matchings that satisfy such notions are better than arbitrary matchings, their main drawback is that they do not take into account agents' preference intensities.  To illustrate this, consider the simple example in \Cref{tab:example} where there are three agents and three objects. The agents all have the same ordinal preferences, but have different preference intensities. For instance, agent 3 only has a slight preference for object 1 over object 2 (0.51 vs 0.49), whereas agents 1 and 2 both prefer object 1 much more than object 2 (0.9 vs 0.1). In this example, if we were concerned about Pareto optimality or rank-maximality, then it is easy to see that any matching satisfies both these notions. However, it is also clear that any matching that matches agent 3 to object 3 is worse in terms of overall welfare, since agent 3 experiences a much larger loss in utility when matched to object 3 instead of object 2, whereas for agents 1 and 2 this difference is very small. 
\begin{table}[t]
	\centering 
	\begin{tabular}{c|c|c|c}
		& object 1 & object 2 & object 3 \\ \hline
		\multicolumn{1}{c|}{agent 1} & 0.9      & 0.1      & 0     \\ 
		\multicolumn{1}{c|}{agent 2} & 0.9      & 0.1      & 0     \\ 
		\multicolumn{1}{c|}{agent 3} & 0.51     & 0.49     & 0        \\ 
	\end{tabular}
	\caption{Example which illustrates how there is a loss in welfare due to not accounting for preference intensities.}
	\label{tab:example}
\end{table}

The observation that there might be a loss in welfare due to not capturing preference intensities (henceforth, cardinal utilities) is not new, and in particular, this has been a much debated issue surrounding various \textit{school-choice} mechanisms (e.g., see \cite{abd11,abd15}).\footnote{School-choice is the problem of allocating students to schools. This is essentially a one-sided matching problem, except that here the schools usually have, what are called, \textit{priorities}---which is essentially an ordering over students that is used only when the school is over-demanded (see, e.g., \cite{abd11} for formal definitions).} This has also lead to, for instance, proposal for new school choice mechanisms that ask agents to provide some extra information along with their ordinal preferences \cite{abd15}. Our work here is partially motivated by this line of work, but takes a more computer science approach to this issue that is similar to the flurry of work that looks at \textit{distortion}---which is essentially the cost of using only ordinal information---in various settings (e.g., see \cite{pro06,bout15,ansh16,ansh17,goel17,abra18}).

In particular, given a one-sided matching instance (which is the set of agents, objects, and the agents' preferences), we are still interested in computing matchings that satisfy some notion of economic efficiency, say, $X$, but at the same time our aim is to account for the agents' cardinal utilities. We do this by aiming to design algorithms that always return a matching that satisfies $X$ and at the same time achieves a good approximation to the optimal welfare among all matchings that satisfy $X$. Now, of course, one way to achieve this is to just ask agents directly for their cardinal utilities, since once we have them we can compute the matching with the best welfare among ones that satisfy $X$. However, as one can imagine, asking agents for their cardinal utilities is cognitively non-trivial, since even attributing values to objects is not easy in many scenarios. Therefore, here our focus is to achieve a middle-ground between completely ordinal and completely cardinal elicitation, and to do this, we use simple threshold queries which ask an agent if their value for an object is at least some real number $v$. Our goal is to ask each agent a small number of such queries and then use it to pick a matching that achieves a good approximation as described above. 

Although the general idea of  using queries to elicit some information regarding cardinal utilities is not new and has even been considered in the context of one-sided matching in a recent paper \cite{aman20b}, there are some differences. The main difference between our approach here and the one by \citet{aman20b} is that the focus in their paper is to come up with algorithms that have low distortion, whereas our focus is on algorithms that, for a particular notion of economic efficiency $X$, always produce a matching that satisfies $X$ and has a good approximation to the optimal welfare among all matchings that satisfy $X$. We believe that while achieving low distortion might be a good objective to be concerned about in certain settings, it is too reliant on cardinal preferences---the presence of which is already an assumption. Hence, comparing different algorithms just based on this value (i.e., the distortion they achieve) does not seem ideal. Our approach, on the other hand, is less reliant on cardinal information, using it only to pick a matching from the set of matchings that satisfy some property $X$, which in turn is dependent only on the ordinal information that is arguably more robust. Moreover, the query model used by \citet{aman20b} is much stronger than the one we employ. Ours just asks for a binary answer to whether the value of an item is greater than some real number $v$, whereas in their model a query asks an agent to reveal their utility for the object, which in turn is cognitively much more demanding. These differences mean that there are no direct overlaps between our work and that of \citet{aman20b}.

\subsection{Our contributions} \label{sec:contribution} 
\begin{table}[t]
	{\small 
		\centering 
		\begin{tabular}{c|c|c|c|}
			\cline{2-4}
			& 
			\begin{tabular}[c]{@{}c@{}}\textbf{Ordinal} \\ \textbf{algorithms} \end{tabular} &
			\begin{tabular}[c]{@{}c@{}}\textbf{Adaptive threshold} \\ \textbf{query algorithms}\\ $\big($for any $\epsilon > 0$, $O(c\log n)$ queries \\per agent, where $c = \ceil{\frac{\log(n^2 \cdot 1/\epsilon)}{\log{(1+\epsilon/2)}}}$ $\big)$ \end{tabular} & \begin{tabular}[c]{@{}c@{}}\textbf{Non-adaptive threshold} \\ \textbf{query algorithms}\\ $\big($at most 1 query \\ per (agent, object) pair$\big)$\end{tabular}                                          \\ \hline
			\multicolumn{1}{|c|}{\begin{tabular}[c]{@{}c@{}}unit-sum \\ valuations\end{tabular}} 
			& \begin{tabular}[c]{@{}c@{}} 
				\rule{0pt}{3ex}  
				UB: $O(n^{2})$ \\
				$[$Theorem~\ref{thm:LB-ordinal}$]$
				\\[1.25ex] 
				LB: $\Omega(n^{2})$ \\
				$[$Theorem~\ref{thm:LB-ordinal}$]$
			\end{tabular}		
			
			&\begin{tabular}[c]{@{}c@{}}
				$1+\epsilon$\\ 
				$[$Theorem~\ref{thm:UB-Adap}$]$
			\end{tabular}                                        
			
			& \begin{tabular}[c]{@{}c@{}} 
				\rule{0pt}{3ex}  
				UB: $O(n^{2/3})$ \\
				$[$Theorems~\ref{thm:UB-nonAdap-PO} and~\ref{thm:UB-nonAdap-priority}$]$\\[1.25ex] 
				LB: $\Omega(\sqrt{n})$ \\
				$[$Theorem~\ref{thm:LB-nonAdap}$]$
			\end{tabular} \\ \hline
			\multicolumn{1}{|c|}{\begin{tabular}[c]{@{}c@{}}unit-range \\ valuations\end{tabular}}  
			& \begin{tabular}[c]{@{}c@{}} 
				\rule{0pt}{3ex}  
				UB: $O(n)$ \\
				$[$Theorem~\ref{thm:LB-ordinal-unitRange}$]$
				\\[1.25ex] 
				LB: $\Omega(n)$ \\
				$[$Theorem~\ref{thm:LB-ordinal-unitRange}$]$
			\end{tabular}
			
			&\begin{tabular}[c]{@{}c@{}}
				$1+\epsilon$\\ 
				$[$Theorem~\ref{thm:UB-Adap}$]$
			\end{tabular}                                                              
			
			& \begin{tabular}[c]{@{}c@{}}
				\rule{0pt}{3ex}
				UB: $O(\sqrt{n})$ \\
				$[$Theorem~\ref{thm:UB-nonAdap-unitRange}$]$\\[1.25ex]
				LB: $\Omega(\sqrt{n})$\\
				$[$Theorem~\ref{thm:LB-nonAdap-unitRange}$]$
			\end{tabular} \\ \hline
		\end{tabular}
		\caption{Summary of our results. For $X$, where $X$ is one of the properties in the set $\{$Pareto optimal, rank-maximal, max-cardinality rank-maximal, fair$\}$, an upper bound (UB) of $\alpha$ indicates that there is a deterministic algorithm that always produces a matching that satisfies $X$ and achieves an $\alpha$-approximation to the optimal welfare among matchings that satisfy $X$. A lower bound (LB) of $\beta$ indicates that there is no deterministic algorithm that produces a matching that satisfies $X$ and achieves a $\beta$-approximation to the optimal welfare among matchings that satisfy $X$.}
		\label{tab:summary}
	}
\end{table}
We consider the following four well-studied types of matchings that satisfy a specific notion of economic efficiency: \textit{i)} Pareto optimal matchings, \textit{ii)} rank-maximal matchings, \textit{iii)} max-cardinality rank-maximal matchings, and \textit{iv)} fair matchings. As mentioned above, for each of these types, our goal is to find deterministic algorithms that always output a matching of the corresponding type and one that achieves a good approximation to the optimal welfare among all matchings of that type. Towards this end, we consider two kinds of cardinal utilities, namely, unit-sum and unit-range valuations, and show the following results, which are summarized in \Cref{tab:summary}.
\begin{itemize}
	\item We first look at adaptive algorithms---i.e., algorithms that are able to change their queries depending on how agents answer its previous queries---and show how for each of the notions mentioned above and for any $\epsilon >0$, there is a deterministic algorithm that asks $O(c \log n)$ queries per agent, where $c = \ceil{\frac{\log(n^2 \cdot 1/\epsilon)}{\log{(1+\epsilon/2)}}}$, and returns a matching that satisfies this notion and also achieves a $(1+\epsilon)$-approximation to the optimal welfare among all matchings that satisfy this notion when the agents have unit-sum or unit-range valuations.
	
	\item While the previous result achieves the best possible approximation one can hope for---and in particular it results in an $O(1)$-approximation with just $O(\log^2 n)$ queries per agent---we believe that the fact it is adaptive is not ideal because of the following reasons: \textit{i)} Adaptive algorithms may not be practical in many settings since it involves waiting for the agents to respond, and potentially having them respond multiple times. \textit{ii)} For every agent, the algorithm mentioned above potentially asks multiple queries with respect to the same object and as this number increases, one could argue that it defeats the real purpose of such algorithms---since responding to them entails that the agents are somewhat sure about their cardinal utilities. 
	
	As a result, we focus on non-adaptive algorithms, which address the first issue mentioned above, and in order to address the second one, we consider a very special type of non-adaptive algorithm---one which is allowed to ask at most one query per (agent, object) pair. We believe that this is the most practical setting to consider for this problem, since this means every agent is asked to provide just one extra bit of information per object.\footnote{Also, deploying this seems easier, since instead of the current system which presumably just asks the agents to list their preferences, now all that needs to be done is to have a checkbox next to it, indicating whether their answer is a ``Yes" or ``No" w.r.t.\ a certain threshold query.}
	
	For this setting, we show how for each of the notions considered, there is a deterministic algorithm that returns a matching that satisfies this notion and also achieves an $O({n}^{2/3})$-approximation to the optimal welfare among all matchings that satisfy this notion when the agents have unit-sum valuations. We also derive a similar result for the unit-range case, showing an algorithm that achieves an $O(\sqrt{n})$-approximation. Note that these bounds are as opposed to a $\Theta({n}^{2})$ and  $\Theta({n})$ approximation that is achievable when using only ordinal preferences and when the agents have unit-sum and unit-range valuations, respectively.   
	
	\item Finally, we also show that, for all the notions considered, any deterministic algorithm that uses at most one query per (agent, object) pair can only achieve an approximation factor of $\Omega(\sqrt{n})$, both for the case when agents have unit-sum or unit-range valuations. Note that for the unit-range case, this bound in turn is asymptotically tight.	
\end{itemize}

\subsection{Related work} \label{sec:rw}   
One-sided matching scenarios are ubiquitous and have been well-studied, especially as the housing allocation or housing market problem both in economics (e.g., see \cite{shapley74,hyll79,roth77,abd98,abd99,son10}) and in theoretical computer science and computational economics (e.g., see \cite{abra04,irv04,abra06,filos14,aman20b}). Unlike in this paper where we assume that the agents have underlying cardinal utilities that are consistent with their ordinal preferences, most of this literature assumes that the agents only have ordinal preferences, and their goal is usually to find matchings that satisfy some notion of economic efficiency like, e.g., Pareto optimal matchings \cite{shapley74,abra04}, rank-maximal matchings \cite{irv04,irv06}, max-cardinality rank-maximal matchings \cite{melh05,abra06}, and fair matchings \cite{melh05,huang13}. 

As mentioned in the Introduction, part of the motivation for our work is derived from work on the school choice problem that talks about the loss in welfare due to not taking the preference intensities into account (e.g., \cite{abd11,abd15}). Our concern here is on similar lines, but we take a more computational approach to the problem which is reminiscent of the vast body of work on distortion (e.g., see \cite{pro06,bout15,ansh16,ansh17,goel17,abra18}). However, unlike this body of work which aims to calculate the worst-case loss in welfare due to only having ordinal preferences, we assume that, in addition to their ordinal preferences, it is also possible to obtain some information about the agents' cardinal utilities. This in turn is similar to the approach taken by  \citet{abra19} in the context of voting and more closely to the ones by \citet{aman20a,aman20b} in the context of voting and one-sided matching, respectively. In the work by \citet{abra19} it assumed that in addition to the ordinal preferences there is also some information regarding how many agents prefer candidate P over Q above a certain threshold, whereas in the work by \citet{aman20a,aman20b}, and as is the case in this paper, it is assumed that the one can use some specific type of query in order to get more information regarding the cardinal utilities. 

Finally, our work is also related to the work that studies the communication complexity of voting protocols \cite{mandal19,mandal20}, to the work on participatory budgeting which compares different elicitation methods based on the distortion achieved (e.g., see \cite{goel19,benade20}), and is more broadly in line with the growing body of work that explicitly aims to make mechanisms or algorithms more robust, by either making use of coarse preference information \cite{chiesa12,chiesa14,menon19} or by making sure that the algorithms designed produce solutions that work ``well" (in the approximation sense) even under slightly modified inputs \cite{shir13,bred17,menon18,mai18,chen19}.
	\section{Preliminaries} \label{sec:prelims}

For $k \in \mathbb{Z}^+$, let $[k]$ denote the set $\{1,\ldots,k\}$. We use $\mn$, where $|\mn| = n$, to denote the set of agents $\{a_1, \ldots, a_n\}$, and use $\mh$, where $|\mh| = n$, to denote the set of objects $\{h_1, \ldots, h_n\}$. We refer to $a_i$ as agent $i$ and $h_j$ as object $j$. Every agent $a_i$ is assumed to have a weak order $P_i$ over a subset of objects $A_i \subseteq \mh$. For an agent $a_i$, $A_i$ indicates the set of objects $a_i$ is willing to be matched to and we refer to $A_i$ as the \textit{acceptable set} of $a_i$ and assume that $|A_i| \geq 1$. We use $\mpp = (P_1, \ldots, P_n)$ to refer to the weak orders of all the agents in $\mn$ and refer to $\mpp$ as the \textit{preference profile} of the agents. For an agent $a_i$, and for two objects $h_j, h_k \in A_i$, we use $h_j \succ_i h_k$ to denote that $a_i$ strictly prefers $h_j$ over $h_k$, and use $h_j \succeq_i h_k$ to indicate that $h_j$ is either strictly preferred or considered to be equivalent to $h_k$. We refer to $\mi = (\mn, \mh, \mpp = (P_1, \ldots, P_n))$ as an \textit{instance}, which encodes all the information about the agents, objects, and the agents' preferences, and use $\mathbb{I}$ to denote the set of all possible instances.

Given an instance $\mi = (\mn, \mh, \mpp = (P_1, \ldots, P_n))$, we use $\mathbb{G}_{\mi} = (\mn \cup \mh, \me)$ to denote the bipartite graph where there as an edge $(a_i, h_j) \in \me$ if $h_j \in A_i$. We refer to  $\mathbb{G}_{\mi}$ as the graph induced by $\mi$ and refer to $e = (a_i, h_j) \in \me$ as a rank-$k$ edge if $|\mathcal{U}_{ij}| = k-1$, where $\mathcal{U}_{ij} = \{h_\ell \in A_i \mid h_\ell \succ_i h_j \}$. We also use $rank(a_i, h_j)$ to denote the $k$ such that $(a_i, h_j)$ is a rank-$k$ edge and refer to an object $h_j$ as $a_i$'s rank-$k$ (or $k$-th choice) object if $rank(a_i, h_j) = k$.

Although the model described thus far is the standard model in one-sided matching, here we additionally assume that each agent $a_i$ has a cardinal utility function $v_i\colon \mh \to [0, 1]$, which is consistent with the preference order $P_i$ (meaning, $h_1 \succeq_i h_2 \Leftrightarrow v_i(h_1) \geq v_i(h_2)$); we assume that if $h \notin A_i$, then $v_i(h) = 0$. In this work we consider two specific kinds of (normalized) valuation functions which are defined below, and use $v = (v_1,\ldots, v_n)$ to denote the valuation profile of agents and $\mathcal{V}_\mi$ to denote the set of all possible valuation profiles that are consistent with the given preference profile in $\mi$. 
\begin{enumerate}
	\item Unit-sum valuations: Agents are said to have unit-sum valuations if for each agent $i$, $v_i$ is such that $\sum_{h\in A_i} v_i(h) = 1$.   
	\item Unit-range valuations: Agents have unit-range valuations if for each agent $i$, there exists $h_j, h_k \in A_i$ such that $h_j \succ_i h_k$, and $\max_{h \in A_i} v_i(h) = 1$ and $\min_{h \in A_i} v_i(h) = 0$. In words, the most preferred objects have value 1, the least preferred objects have value 0, and every other acceptable object has value between 0 and 1. 
\end{enumerate}
Note that information about the cardinal utilities is not part of an instance $\mi$. Also, note that although the ``internal utilities" of the agents may not be normalized, we have to assume some normalization in order to only use threshold queries. Given this, both unit-sum and unit-range arise from two natural ways to normalize agents' ``internal utilities". More precisely, if $u_i\colon \mh \to \mathbb{R}_{\geq 0}$ is the ``internal utility" of an agent, then, for all $h \in _iA$, the corresponding unit-sum and unit-range valuation functions are, respectively, 
\[
v_i(h) = \frac{u_i(h)}{\sum_{h\in A_i} u_i(h)}
\qquad \text{and} \qquad
v_i(h) = \frac{u_i(h) - \min_{h \in A_i} u_i(h)}{\max_{h \in A_i} u_i(h) - \min_{h \in A_i} u_i(h)}.
\]

For an instance $\mi$, we are interested in matchings that assign agents to objects, and a \textit{matching} of agents to objects is a bijection $\mu\colon \mn \to \mh$ and, for $c \in \mn \cup \mh$, we refer to $\mu(c)$ as $c$'s \textit{partner} in $\mu$ or as $c$'s \textit{allocation} in $\mu$. Alternatively, a matching is also defined as a collection of edges $\mu$ in $\bbg_{\mi}$ such that each vertex is part of at most one edge in $\mu$. We use $\mm_{\bbg_{\mi}}$ to denote the set of all possible matchings in $\bbg_{\mi}$. Although for a given instance there are several possible matchings, we are interested in matchings that satisfy some notion of economic efficiency; these are defined next.

\subsection{Notions of economic efficiency} \label{sec:typesOfmatchings}
We consider the following well-studied notions: Pareto optimal matchings \cite{shapley74,abra04}, rank-maximal matchings \cite{irv04,irv06}, max-cardinality rank-maximal matchings \cite{melh05,abra06}, and fair matchings \cite{melh05,huang13}.  The latter three are different ways to strengthen Pareto optimality and are together referred in the rest of this paper as \textit{signature-based} matchings.

\begin{definition}[Pareto optimal matchings]
	Given an instance $\mi = (\mn, \mh, \allowbreak \mpp = (P_1, \allowbreak  \ldots, \allowbreak  P_n))$, a matching $\mu \in \mm_\mathbb{G_\mi}$ is Pareto optimal (PO) w.r.t.\ $\mi$ if  
	\begin{equation*}
		\forall \mu' \in \mm_\mathbb{G_\mi}: \; \left(\exists a_i \in \mn, \mu'(a_i) \succ_{i} \mu(a_i) \right) \Rightarrow \left(\exists a_j \in \mn, \mu'(a_j) \prec_{j} \mu(a_j)\right).  
	\end{equation*}
\end{definition}
In words, a matching $\mu$ is Pareto optimal if there is no other matching $\mu'$ such that
every agent weakly-prefers their allocation in $\mu'$ over their allocation in $\mu$ and at least one agent (strictly) prefers their allocation in $\mu'$ over their allocation in $\mu$. 

\begin{definition}[Signature-based matchings]
	Given an instance $\mi = (\mn, \mh, \mpp = (P_1, \ldots,\allowbreak  P_n))$, and a matching $\mu \in \mm_\mathbb{G_\mi}$, let $s_i$ denote the number of agents that are matched to a rank-$i$ edge in $\mu$. Then, $\mu$ is  
	\begin{itemize}
		\item \textbf{rank-maximal} if $\mu$ maximizes the number of agents who are matched to a rank-1 edge and, subject to that, it maximizes the number of agents who are matched to rank-2 edges, and so on. Formally, if we associate an $n$-tuple $(s_1, \ldots, s_n)$ with every matching in $\mm_\mathbb{G_\mi}$, then $\mu$ is the matching that has lexicographically the best $n$-tuple $(s_1, \ldots, s_n)$ associated with it. 
		
		\item \textbf{max-cardinality rank-maximal} if $\mu$ is a maximum cardinality matching and, subject to that, is also rank-maximal. Formally, if we associate an $(n+1)$-tuple $(\sum_{i=1}^{n} s_i, s_1, \ldots, s_n)$ with every matching in $\mm_\mathbb{G_\mi}$, then $\mu$ is the matching that has lexicographically the best $(n+1)$-tuple $(\sum_{i=1}^{n} s_i, s_1, \ldots, s_n)$ associated with it. 
		
		\item \textbf{fair} if $\mu$ is a maximum cardinality matching and, subject to that, minimizes the number of agents who are matched to a rank-$n$ edge and, subject to that, minimizes the number of agents who are matched to a rank-$(n-1)$ edge, and so on. Formally, if we associate an $(n+1)$-tuple $(\sum_{i=1}^{n} s_i, -s_n, -s_{n-1}, \ldots, -s_1)$ with every matching in $\mm_\mathbb{G_\mi}$, then $\mu$ is the matching that has lexicographically the best $(n+1)$-tuple $(\sum_{i=1}^{n} s_i, -s_n, -s_{n-1}, \ldots, -s_1)$ associated with it.	
	\end{itemize}
\end{definition}
For each type of signature-based matching defined above and a matching of that type, we refer to the corresponding tuple, as defined above, to be the matching's signature. That is, for instance, for a max-cardinality rank-maximal matching $\mu$, signature of $\mu$ refers to the $(n+1)$-tuple $(\sum_{i=1}^{n} s_i, s_1, \ldots, s_n)$. The example below illustrates the difference between Pareto optimal, rank-maximal, and fair matchings for an instance.
\begin{example}
	Consider an instance with $7$ agents and $7$ objects, where the preferences of the agents are as defined in the table below. Here each column corresponds to the strict preferences of an agent and if the column corresponding to agent, say, $a_i$ is $h_1, h_2, h_3$, then this implies that $a_i$ prefers $h_1$ the most, $h_2$ second most, and so on.   
	\begin{table}[ht]
		\centering 
		\begin{tabular}{c|c|c|c|c|c|c}
			{\color[HTML]{000000} $a_1$}                         & {\color[HTML]{000000} $a_2$}                         & {\color[HTML]{000000} $a_3$}                         & {\color[HTML]{000000} $a_4$}                         & {\color[HTML]{000000} $a_5$}                         & {\color[HTML]{000000} $a_6$}                         & {\color[HTML]{000000} $a_7$}                         \\ \hline
			{\color[HTML]{000000} $h_1$}                         & \cellcolor[HTML]{FD6864}{\color[HTML]{000000} $h_2$} & \cellcolor[HTML]{CBCEFB}{\color[HTML]{000000} $h_1$} & \cellcolor[HTML]{FD6864}{\color[HTML]{000000} $h_3$} & {\color[HTML]{000000} $h_1$}                         & \cellcolor[HTML]{FD6864}{\color[HTML]{000000} $h_1$} & \cellcolor[HTML]{FFC702}{\color[HTML]{000000} $h_1$} \\
			\rowcolor[HTML]{FFC702} 
			\cellcolor[HTML]{FD6864}{\color[HTML]{000000} $h_4$} & {\color[HTML]{000000} $h_5$}                         & {\color[HTML]{000000} $h_3$}                         & {\color[HTML]{000000} $h_6$}                         & {\color[HTML]{000000} $h_4$}                         & {\color[HTML]{000000} $h_2$}                         & \cellcolor[HTML]{CBCEFB}{\color[HTML]{000000} $h_2$} \\
			\cellcolor[HTML]{CBCEFB}{\color[HTML]{000000} $h_3$} & \cellcolor[HTML]{CBCEFB}{\color[HTML]{000000} $h_6$} & {\color[HTML]{000000} }                              & {\color[HTML]{000000} }                              & \cellcolor[HTML]{CBCEFB}{\color[HTML]{000000} $h_5$} & \cellcolor[HTML]{CBCEFB}{\color[HTML]{000000} $h_4$} & \cellcolor[HTML]{FD6864}{\color[HTML]{000000} $h_5$} \\
			\cellcolor[HTML]{FFCB2F}{\color[HTML]{000000} $h_7$} & {\color[HTML]{000000} }                              & {\color[HTML]{000000} }                              & {\color[HTML]{000000} }                              & {\color[HTML]{000000} }                              & {\color[HTML]{000000} }                              & {\color[HTML]{000000} }                              \\ \hline
		\end{tabular}
		\label{tab:my-table}
	\end{table}
	For this instance, the matching (of size 6) that corresponds to each agent being matched to the object (if any) coloured blue in its column is an example of a Pareto optimal matching. Similarly, the matching (of size 5, with signature $(3,1,1,0,0,0)$) where each agent is matched to the object (if any) coloured red is an example of a rank-maximal matching, and the matching (of size 7, with signature $(7,0,0,-1,0,-5, -1)$) where each is matched to the object (if any) coloured yellow is an example of a fair matching.
\end{example}
It is well-known (see \cite{irv04, irv06, melh05, huang13, michail07}) that signature-based matchings can be reduced to an instance of the following problem, which we refer to as priority-$\bmp$ matchings, for a given $\bmp = (p_1, \cdots, p_n)$.

\begin{definition} \label{def:priorityp}
	Given an instance $\mi = (\mn, \mh, \mpp = (P_1, \ldots, P_n))$ and a priority vector $\bmp = (p_1,\allowbreak\ldots, p_n)$, where $\forall i \in [n], p_i \in \mathbb{Z}_{\geq0}$ and $\exists j, k \in [n]$ such that $p_j \neq p_k$, a matching $\mu \in \mm_\mathbb{G_\mi}$ is said to be  a priority-$\bmp$ matching if $\mu$ is a matching of maximum weight in $\mm_\mathbb{G_\mi}$, where a rank-$r$ edge in $\bbg_\mi$ is assigned the weight $p_r$. 
\end{definition}

In particular, given an instance $\mi$, we can show that,\footnote{The proof of this can be found in \Cref{clm:sigTopriority} in \Cref{app:addClaims}. Note that such observations have also been made in previous works (e.g., see \cite{irv04, irv06, melh05, huang13, michail07}), although the value of $p_j$s used may be different.}  
\begin{itemize}
	\item when $p_j = n^{2(n-j+1)}$ for all $j \in [n]$, a matching is a priority-$\bmp$ matching if and only if it is rank-maximal matching w.r.t.\ $\mi$.
	\item when  $p_j = n^{2n} + n^{2(n-j)}$ for all $j \in [n]$, a matching is a priority-$\bmp$ matching if and only if it is a max-cardinality rank-maximal matching w.r.t.\ $\mi$.
	\item when  $p_j = 4n^{2n} - 2n^{j-1}$ for all $j \in [n]$, a matching is a priority-$\bmp$ matching if and only if it is a fair matching w.r.t.\ $\mi$.
\end{itemize}	

Although priorty-$\bmp$ matchings can potentially be defined for several values of $\bmp$, here we are interested in the three cases described above. Additionally, for ease of exposition, we also sometimes use priority-$\bmp$, where $p_i = 0$ for all $i \in [n]$ to refer to Pareto optimal matchings. Note that this is purely for notational convenience (priority-$\bmp$ matching as defined in \Cref{def:priorityp} is not defined when $p_i = p_j$ for all $i,j \in [n]$), since the algorithms we discuss in the context of Pareto optimal matchings are extensions to the ones for priority-$\bmp$ matchings. Throughout, we use $\mathbbm{P}$ to denote the priority vectors of interest. That is, $\mathbbm{P} = \{(n^{2n}, \dots, n^{2(n-j+1)}, \ldots, n^{2}), \allowbreak(n^{2n} + n^{2(n-1)}, \allowbreak \ldots, \allowbreak n^{2n} + n^{2(n-j)},\allowbreak\ldots, n^{2n} + n),\allowbreak (4n^{2n} - 2, \allowbreak \ldots, 4n^{2n} - 2n^{j-1},\allowbreak\ldots, 4n^{2n} - 2n^{n-1})\}$. 

\subsection{Going beyond completely ordinal or completely cardinal algorithms}
Given an instance $\mi$, we are interested in deterministic algorithms $\ma$ that always output a matching that satisfies one of the economic notions defined in the previous section. However, even when restricted to such matchings, as the example in the Introduction illustrates, there are potentially many choices, and there might be a loss in welfare due to not accounting for the cardinal utilities of the agents. Therefore, ideally we  want our algorithm to have small worst-case loss in welfare. Formally, for an instance $\mi$, consider the set of matchings $S \subseteq \mm_{\bbg_{\mi}}$ such that $S$ is the set of all Pareto optimal/rank-maximal/max-cardinality rank-maximal/fair matchings in $\bbg_{\mi}$. Next, for a matching $\mu \in S$, $v \in\mathcal{V}_{\mi}$, and for an edge $e = (a_i, h_j) \in \mu$, let $\val(e) = v_i(h_j)$ and $\sw(\mu \mid v) = \sum_{e \in \mu} \val(e)$, the social welfare of $\mu$ given the valuations $v$. (For notational convenience, when $v$ is clear from the context, we just write $\sw(\mu)$ instead of $\sw(\mu \mid v)$.) Given this, consider a deterministic algorithm $\ma$ where, for all $\mi \in \mathbb{I}$, $\ma(\mi) \in S$ and let $\ml(\ma)$, which we refer to as the \textit{worst-case welfare loss} of $\ma$, be defined as below.
\begin{align} \label{def:l(a)}
	\mathcal{L}(\ma) &= \max_{\mi \in \mathbb{I}} \ml(\ma, \mi), \text{ where }\mathcal{L}(\ma, \mi) \: = \: \sup_{v \in \mathcal{V}_\mi} \: \frac{\displaystyle\max_{\mu^* \in S} \: {\sw(\mu^* \mid v)}}{\sw(\ma(\mi) \mid v)}.
\end{align}  

As mentioned above, we want algorithms $\ma$ that have as small a value of $\ml(\ma)$ as possible. To achieve this, on the one extreme we have completely ordinal algorithms---which are algorithms that only consider the ordinal preferences. We argue below that that any deterministic ordinal algorithm has a very poor worst-case loss in welfare, in particular $\Omega(n^2)$ when agents have unit-sum valuations. The proof of this result appears in \Cref{app:sec:ordinal}. (The corresponding result showing a bound of $\Omega(n)$ for unit-range valuations can be found as \Cref{thm:LB-ordinal-unitRange} in \Cref{app:sec:ordinal}.) 
\begin{restatable}{theorem}{LBordinal}  \label{thm:LB-ordinal}
	Let $X$ denote one of the properties in the set $\{${Pareto-optimal, $\allowbreak$ rank-maximal, max-cardinality rank-maximal, and fair}$\}$. Let $\ma$ be a deterministic ordinal algorithm that always produces a matching that satisfies property $X$. If there are $n$ agents with unit-sum valuation functions, then $\ml(\ma) \in \Omega({n}^2)$. Moreover, this bound is asymptotically tight.
\end{restatable}

At the other extreme, when we have access to all the cardinal utilities, we show (see \Cref{thm:computeOptPriority}) in the next section how, for all the notions considered here, computing the welfare-optimal matching reduces to the max-weight matching problem. Although this is ideal, as mentioned in the Introduction, asking agents for cardinal utilities might not be reasonable in many situations, as this is a cognitively involved task. Therefore, in this paper, we aim for a middle-ground between completely ordinal and completely cardinal algorithms. We do this by trying to get at least some information regarding the cardinal preferences by asking the agents certain queries. In particular, we are interested in the following type of query, which we refer to as \textit{binary threshold query}. 
\begin{definition}[binary threshold query]
	For an agent $a_i$, object $h_j$, and a real number $t_k \in [0, 1]$, a binary threshold query, denoted $\mq(a_i, h_j, t_k)$, asks agent $a_i$ to return 1 (alternatively, asks them to say ``Yes") if $v_i(h_j) \geq t_k$, and 0 (alternatively, asks them to say ``No") otherwise.
\end{definition}

Given an instance $\mi$ and answers to a certain number of binary threshold queries, our aim is to design deterministic algorithms $\ma$ that minimize the worst-case welfare loss $\ml(\ma)$ and, for all $\mi \in \mathbb{I}$, produces a matching in $S$ (i.e., $\ma(\mi) \in S$), where $S$ is the set of all Pareto optimal/rank-maximal/max-cardinality rank-maximal/fair matchings in $\bbg_{\mi}$.

\textbf{Remark:} Throughout this paper, we say that, for an $\alpha \geq 1$, an algorithm $\ma$ achieves an $\alpha$-approximation to the optimal social welfare among Pareto-optimal/rank-maximal/max-cardinality rank-maximal/fair matchings if $\mathcal{L}(\ma) \leq \alpha$. Also, note that although the ratio defined in (\ref{def:l(a)}) might seem very similar to the notion of \textit{distortion} that is widely used in computational social choice literature (e.g., see \cite{pro06,ansh16}), it is important to note that it is different. Here we are interested in algorithms that produce a matching with a certain property (like Pareto optimality, rank-maximality, etc.) and that has social welfare as close to the optimal welfare achievable with a matching that satisfies the same property of interest, whereas in the context of distortion there is no such restriction. That said, it is also worth noting that in the context of one-sided matching, an algorithm that achieves an $\alpha$-approximation to the optimal social welfare among Pareto-optimal matchings also has a distortion of $\alpha$ since a welfare-optimal matching is also Pareto optimal.

\subsection{Finding welfare-optimal priority-\texorpdfstring{$\bmp$}{p} matchings when utilities are known} \label{sec:prelims:computeOpt}
Before we consider the main question of this paper, a natural question that arises is on how to compute the welfare-optimal Pareto optimal or welfare-optimal priority-$\bmp$ matchings for the priority vectors of interest, when the agents' utilities are known. Given an instance $\mi = (\mathcal{N}, \mathcal{H}, \mpp)$ and valuation functions of the agents $v = (v_1, \ldots, v_n)$, where $v_i\colon \mh \to [0, 1]$, the welfare-optimal priority-$\bmp$ problem is to find a matching of maximum welfare among the set of priority-$\bmp$ matchings. Below, we show how for all the priority vectors of interest, computing this reduces to an instance of the max-weight matching problem on $\mathbb{G}_{\mi}$. 

To see this, first note that finding the welfare-optimal Pareto optimal matching directly reduces to the max-weight matching problem on $\mathbb{G}_{\mi}$, where the weight of an edge $(a_i, h_j)$ is $v_i(h_j)$. Therefore, below we show how even the question of computing the welfare-optimal priority-$\bmp$ matchings for $\bmp \in \mathbb{P}$ reduces to the max-weight matching problem on $\mathbb{G}_{\mi}$.

\begin{algorithm}[t]
	{\small
		\centering
		\noindent\fbox{%
			\begin{varwidth}{\dimexpr\linewidth-4\fboxsep-4\fboxrule\relax}
				\begin{algorithmic}[1]
					\small 
					\Input an instance $\mi = (\mn, \mh, \mpp = (P_1, \ldots, P_n))$, priorities $\bmp = (p_1, \ldots, p_n)$, where $\bmp \in \mathbbm{P}$, and $v = (v_1, \ldots, v_n)$, where $v_i\colon\mh \to [0, 1]$ 
					\Output returns a welfare-optimal priority-$\bmp$ matching w.r.t.\ $\mi$
					
					\State $\mathbb{G}_{\mi} = (\mn \cup \mh, \me) \gets$ graph induced by $\mi$
					\For{$e = (a_i, h_j) \in \me$} 
					\State $r \gets \text{rank}(a_i, h_j)$
					\State $w_e \gets p_r + v_i(h_j)$ \label{algo:line:weights}
					\EndFor
					\State $\mu \gets$ max-weight matching in $\mathbb{G}_{\mi}$ with weights $\{w_e\}_{e \in \me}$ \label{line:Gdef}
					\State \Return $\mu$			
				\end{algorithmic}
		\end{varwidth}}
		\caption{returns a welfare-optimal priority-$\bmp$ matching.}
		\label{algo:computeOptPriority}
	}
\end{algorithm}

\begin{theorem} \label{thm:computeOptPriority}
	Given an instance $\mi = (\mathcal{N}, \mathcal{H}, \mpp = (P_1, \ldots, P_n))$, a vector of priorities $\bmp = (p_1, \allowbreak \ldots, p_n)$, where $\bmp \in \mathbbm{P}$, and $v = (v_1, \ldots, v_n)$, where $v_i\colon \mh \to [0, 1]$, \Cref{algo:computeOptPriority} returns a welfare-optimal  priority-$\bmp$ matching w.r.t.\ $\mi$.
\end{theorem}
\begin{proof} 
	Let $\mu$ be the matching that is returned by \Cref{algo:computeOptPriority}. First, we argue that $\mu$ is a priority-$\bmp$ matching w.r.t.\ $\mi$. To do this, fix a $\bmp \in \mathbbm{P}$, and let us suppose that $\mu$ is not a priority-$\bmp$ matching, but is a max-weight matching in $\mathbb{G}_{\mi}$ with weights $\{w_e\}_{e \in \me}$, where for $e = (a_i, h_j)$ and $r(e) = rank(a_i, h_j)$, $w_e = p_{r(e)} + v_i(h_j)$ (see line~\ref{algo:line:weights} in \Cref{algo:computeOptPriority}). Let $\mu'$ be a priority-$\bmp$ matching w.r.t.\ $\mi$, which by definition means that it is the max-weight matching in $\mathbb{G}_{\mi}$ with weights $\{w'_e\}_{e \in \me}$, where $w'_e = p_{r(e)}$. Also, for a matching $\mu_1 \in \mm_{\mathbb{G}_\mi}$, let $W[\mu_1]$ and $W'[\mu_1]$ denote the sum of the edge weights in $\mu_1$ when using weights $\{w_e\}_{e \in \me}$ and $\{w'_e\}_{e \in \me}$, respectively. 
	
	Next, note that for any matching $\mu_1 \in \mm_{\mathbb{G}_\mi}$, $W'[\mu_1] \leq W[\mu_1]$. Additionally, note that, since $p_j \geq  n^2$ for all $j \in [n]$, $W'[\mu'] - W'[\mu] \geq n^2$. Finally, observe that, $W[\mu] = \sum_{e \in \mu} w(e) = \sum_{e \in \mu} p_{r(e)} + v_i(h) \leq W'[\mu] + n < W'[\mu'] \leq W[\mu']$, where the first inequality follows from the fact that $v_i(\cdot) \in [0, 1]$ and the second inequality follows since $W'[\mu'] - W'[\mu] \geq n^2$. However, note that $W[\mu] < W[\mu']$ is a contradiction since $\mu$ is the max-weight matching in $\mathbb{G}_{\mi}$ with weights $\{w_e\}_{e \in \me}$.
	
	Given that $\mu$ is a priority-$\bmp$ matching, the fact that it is welfare-optimal follows since we are computing the max-weight matching in $\mathbb{G}_{\mi}$ and any priority-$\bmp$ matching $\mu' \in \mm_{\mathbb{G}_{\mi}}$.
\end{proof}
	\section{Improving Welfare using Threshold Queries}
In this section we look at the main question considered in this paper, which is broadly: \textit{how can one improve social welfare in one-sided matching problems by asking only a small number of queries regarding cardinal utilities}. As mentioned previously, we are interested in binary threshold queries, $\mq(a_i, h_j, t_k)$, which asks an agent $a_i$ if a particular object $h_j$ is of value at least $t_k$. Towards this end, we begin by considering adaptive algorithms---i.e., algorithms that are allowed to change its queries based on the agents' responses to its previous queries---and show how, when considering each of the four notions (i.e., Pareto optimal, rank-maximal, max-cardinality rank-maximal, and fair matchings) of interest, one can obtain a $(1+\epsilon)$-approximation to the optimal welfare among all matchings that satisfy that notion. Following this, we look at non-adaptive algorithms, which we believe are the more interesting and practical ones for this setting. In particular, we restrict ourselves to algorithms that can ask at most one query per (agent, object) pair  and show upper and lower bounds on the approximation achievable. Unless explicitly specified, the results in this section work with respect to both unit-sum and unit-range valuations.

\subsection{Adaptive algorithm to achieve \texorpdfstring{$(1+\epsilon)$}{}-approximation} \label{sec:Adap}
The idea in \Cref{algo:adaptiveApproxMaxweight} is simple. For a specific choice of $c$, it associates a partition of objects with respect to every agent, where, for $k\in[c]$, an object is in $\me_{ik}$ if agent $i$'s value of for the object is within the interval $B_k = [t_k, t_{k-1})$, where $t_k = (\frac{2}{2+\epsilon})^k$, $t_{0} = 1$, and the right endpoint interval is closed when $k = 1$. Following this, for every edge $e = (a_i, h_j)$ in $\mathbb{G}_{\mi}$, it assigns a weight $w_e = p_r + t_k$, where $r = rank(a_i, h_j)$, $h_j \in \me_{ik}$, and $t_k$ is the left endpoint of the interval $B_k$, and computes the max-weight matching on the resulting weighted graph. Below we show for $c = \ceil{\frac{\log(n^2 \cdot 1/\epsilon)}{\log{(1+\epsilon/2)}}}$, this results in an $(1+\epsilon)$-approximation algorithm that uses $O(c \log n)$ queries per agent. In particular, this means that one can achieve a 2-approximation using $O(\log^2 n)$ queries per agent.
\begin{algorithm}[tb]
	{\small 
		\centering
		\noindent\fbox{%
			\begin{varwidth}{\dimexpr\linewidth-4\fboxsep-4\fboxrule\relax}
				\begin{algorithmic}[1]
					\small 
					\Input an $\epsilon > 0$, an instance $\mi = (\mn, \mh, \mpp)$, and a priority vector $\bmp = (p_1, \ldots, p_n)$ 
					\Output returns a PO matching when $p_i=0$ for all $i \in [n]$ and a priority-$\bmp$ matching when $\bmp \in \mathbb{P}$
					
					\State $\mathbb{G}_{\mi} = (\mn \cup \mh, \me) \gets$ graph induced by $\mi$
					\State $c \gets \left\lceil \frac{\log(n^2 \cdot 1/\epsilon)}{\log{(1+\epsilon/2)}} \right \rceil$
					\State $t_i \gets (\frac{2}{2+\epsilon})^i$, for $i \in [c]$ 
					
					\For{$a_i \in \mn$}
					\For{$k \in [c]$}						
					\State $\me_{ik} \gets \{(a_i, h_j) \in \me \mid \mq(a_i, h_j, t_k) = 1 \text{ and, if } k\geq 2, \mq(a_i, h_j, t_{k-1}) = 0\}$  \label{algo:line:eik}
					\For{$e = (a_i, h_j) \in \me_{ik}$} 
					\State $r \gets \text{rank}(a_i, h_j)$
					\State $\val'(e) \gets t_k$  \label{algo:line:val'}
					\State $w_e \gets p_r + \val'(e)$ \label{algo:assign-weight}
					\EndFor	
					\EndFor						
					\EndFor
					\State $\mu \gets$ max-weight matching in $\mathbb{G}_{\mi}$ with weights $\{w_e\}_{e \in \me}$ \label{algo:line:muGen}
					\If{$p_i = 0$ for all $i \in [n]$} 
					\State $\mu \gets$ run top-trading cycles (TTC) algorithm  with $\mu$ (from line~\ref{algo:line:muGen}) as the initial endowment		\label{algo:line:ttc}		
					\EndIf
					\State \Return $\mu$ 
				\end{algorithmic}
		\end{varwidth}}
		\caption{returns a PO matching that achieves a $(1+\epsilon)$-approximation to the optimal social welfare among PO matchings or a priority-$\bmp$ matching that achieves a $(1+\epsilon)$-approximation to the optimal social welfare among priority-$p$ matchings.}
		\label{algo:adaptiveApproxMaxweight}
	}
\end{algorithm}

\begin{theorem} \label{thm:UB-Adap}
	Given an $\epsilon > 0$, an instance $\mi = (\mn, \mh, \mpp)$, and a priority vector $\bmp = (p_1, \ldots, p_n)$, \Cref{algo:adaptiveApproxMaxweight} is an adaptive algorithm that asks $O(c \log n)$ queries per agent, where $c = \ceil{\frac{\log(n^2 \cdot 1/\epsilon)}{\log{(1+\epsilon/2)}}}$, and returns a
	\begin{enumerate}[label=\roman*),ref=\roman*]
		\item Pareto optimal matching $\mu$ that achieves a $(1+\epsilon)$-approximation to the optimal welfare among all Pareto optimal matchings when $p_i = 0$ for all $i \in [n]$.
		\item priority-$\bmp$ matching $\mu$ that achieves a $(1+\epsilon)$-approximation to the optimal welfare among all priority-$\bmp$ matchings when $\bmp \in \mathbb{P}$.
	\end{enumerate}   
\end{theorem}

\begin{proof}We first argue that \Cref{algo:adaptiveApproxMaxweight} adaptively asks $O(c \log n)$ queries per agent. To see this, note that for each agent $a_i$ and any $k \in [c]$, creating the set $\me_{ik}$ takes $O(\log |A_i|)$ queries. This is so because we know $a_i$'s weak order over $A_i$, the set of acceptable objects of $a_i$, and because we are using an adaptive algorithm, we can see the results of one query before we ask the next. Therefore, we can perform binary search to determine the least preferred object  $h_{j_k}$ such that $v_i(h_{j_k}) \in [t_k, t_{k-1})$, where $t_i = (\frac{2}{2+\epsilon})^i$, $t_{0} = 1$, and the right endpoint interval is closed when $k = 1$ (note that there could multiple such objects since we allow ties, but we can still do this since we have the weak order). Additionally, since, for all $k \in [c]$, $t_k < t_{k-1}$, it is easy to see that  we can start from $k=1$ and proceed to $k=c$ and form the sets $\me_{ik} = \{h \in \mh \mid h \succeq h_{j_k} \text{ and, if } k\geq 2, h \prec h_{j_{k-1}}\}$. Therefore, since we create $c$ sets per agent, \Cref{algo:adaptiveApproxMaxweight}  asks $O(c \log n)$ queries per agent.
	
	Next, we need to argue that the matching returned is Pareto optimal when $p_i = 0$ for all $i \in [n]$ and is a priority-$\bmp$ matching when $\bmp \in \mathbb{P}$. To see this, note that in the former case, the matching returned is one that is returned by the top-trading cycles (TTC) algorithm (see line~\ref{algo:line:ttc} in \Cref{algo:adaptiveApproxMaxweight}). Therefore, we know that it is PO \cite{shapley74}. As for the case when $\bmp \in \mathbb{P}$, note that the matching returned is the max-weight matching in $\mathbb{G}_{\mi}$ with weights $\{w_e\}_{e \in \me}$, where $w_e = p_r + \val'(e)$ (see line~\ref{algo:assign-weight}). Now, since for any $e \in \me$, $\val'(e) \in [0, 1]$, we know from \Cref{thm:computeOptPriority} that such a matching is a priority-$\bmp$ matching.
	
	Finally, in order to show that the returned matching achieves a $(1+\epsilon)$-approximation to optimal welfare among all Pareto optimal/priority-$\bmp$ matchings, let $\mu^*$ be the matching that maximizes the total welfare when true edge-weights (which in turn are the true utilities of the agents) are known. That is, $\mu^*$ is the max-weight matching in $\mathbb{G}_{\mi}$ with weights $(w^*_e)_{e \in \me}$, where, for $e = (a_i, h_j) \in \me$, $w^*_e = p_r + v_i(h_j)$. By \Cref{thm:computeOptPriority} we know that $\mu^*$ is the welfare-optimal priority-$\bmp$ matching when $\bmp \in \mathbb{P}$ and it is the welfare optimal Pareto optimal matching when $p_i = 0$ for all $i\in[n]$. So, now, if $\mu$ is the matching computed by \Cref{algo:adaptiveApproxMaxweight}, then we need to show that $\sw(\mu^*) \leq (1+\epsilon) \cdot \sw(\mu)$. 
	
	To see this, let us partition the edges in $\mu^*$ into sets $H$ and $L$ such that $H = \{(a_i, h_j) \in \mu^* \mid v_i(h_j) \geq \frac{\epsilon}{2n^2}\}$ and $L = \mu^* \setminus H$. Next, note that we have chosen the value of $c$ such that $t_c \geq \frac{\epsilon}{2n^2}$. This in turn implies that for any $k \in [c]$ and $e \in \me_{ik}$, the value assigned to $e = (a_i, h_j) \in H$ in the algorithm (i.e, $\val'(e)$ assigned in line~\ref{algo:line:val'} in \Cref{algo:adaptiveApproxMaxweight}) is at least $\frac{t_{k}}{t_{k-1}} \cdot \val(e)$ (and is at most $\val(e)$). So, using this, we have that $\sw(H) = \sum_{e \in H}\val(e) \leq \frac{2+\epsilon}{2}\sum_{e \in H}\val'(e) \leq \frac{2+\epsilon}{2}\cdot \sw(\mu')$, where $\mu'$ is the matching computed in line~\ref{algo:line:muGen} and the last inequality follows from the fact that $H \subseteq \me$, and $\mu'$ is the max-weight matching in $\mathbb{G}_{\mi}$ (see line~\ref{algo:line:muGen}).
	%
	
	Finally, in order to get the approximation bound, let us first bound the ratio $\frac{\sw(\mu^*)}{\sw(\mu')}$. Note that since $|L| \leq n$ and for each $e = (a_i, h_j) \in L$, $v_i(h_j) \leq \frac{\epsilon}{2n^2}$, we have,	  
	\begin{align} \label{eqn:adap}
		\frac{\sw(\mu^*)}{\sw(\mu')} = \frac{\sw(H) + \sw(L)}{\sw(\mu')}
		\leq \frac{\frac{2+\epsilon}{2} \cdot \sw(\mu') + n \cdot \frac{\epsilon}{2n^2}}{\sw(\mu')}.
	\end{align}
	
	Next, note that for the case when $\bmp \in \mathbb{P}$, $\mu' = \mu$, and since $\mu'$ is a priority-$\bmp$ matching, we have that $\sw(\mu') = \sw(\mu) \geq \frac{1}{n}$, where the inequality follows since in a priority-$\bmp$ matching at least one of the agents is matched to a rank-1 object (see Claim~\ref{clm:boundB} in \Cref{app:addClaims} for a proof) and the valuations are either unit-sum or unit-range. Therefore, using this along with (\ref{eqn:adap}), we have that, 
	\begin{align*}
		\frac{\sw(\mu^*)}{\sw(\mu')}
		\leq \frac{\frac{2+\epsilon}{2} \cdot \sw(\mu') + n \cdot \frac{\epsilon}{2n^2}}{\sw(\mu')}
		&\leq 1 + \epsilon.
	\end{align*}
	
	For the case when $p_i = 0$ for all $i\in[n]$, note that $\mu$ is the matching that is returned by TTC on initial endowments $\mu'$. Therefore, since TTC is individually rational and returns a matching that is PO \cite{shapley74}, we have that $\sw(\mu) \geq \sw(\mu')$. Hence, using this along with (\ref{eqn:adap}), we have that,
	\begin{align*}
		\frac{\sw(\mu^*)}{\sw(\mu)} \leq \frac{\frac{2+\epsilon}{2} \cdot \sw(\mu') + n \cdot \frac{\epsilon}{2n^2}}{\sw(\mu)}
		&\leq 1 + \epsilon.
	\end{align*}
	where the last inequality follows from the fact that $\sw(\mu) \geq \frac{1}{n}$, since in a Pareto optimal matching at least one of the agents is matched to a rank-1 object (see Claim~\ref{clm:boundB} in \Cref{app:addClaims} for a proof) and the valuations are either unit-sum or unit-range.
\end{proof}

\textbf{Remark:} As stated previously, an algorithm that produces a PO matching and also achieves an $\alpha$-approximation to the optimal welfare among all PO matchings has a distortion of $\alpha$, since a welfare-optimal matching is also PO. This in turn implies that, although the objectives here and in the paper by \citet{aman20b} are different, the distortion guarantees implied by the algorithm above is similar to the one provided by the $\lambda$-ThresholdStepFunction algorithm in their paper (which, though, even works for unrestricted valuations, unlike the case here). 
Moreover, while the algorithms share some similarities, we use a weaker query model and, most importantly, for an appropriate choice of $\bmp$, our algorithm produces a Pareto optimal or a priority-$\bmp$ matching.

\subsection{Non-adaptive algorithms: asking one query per (agent, object) pair} \label{sec:nonAdap}
In this section we turn our attention to non-adaptive algorithms, in particular looking at algorithms that can only ask one query per (agent, object) pair and cannot change these queries depending on the responses to previous ones. As mentioned in the Introduction, we believe that, at least in some contexts, this is the more interesting and practical setting to consider, since such an algorithm does not have to wait for the agents to respond and also does not require an agent to answer multiple queries with respect to the same object---doing which would in turn entail that the agent is somewhat sure about their cardinal utilities.

Below, we present two algorithms for when agents have unit-sum valuations, first in the context of priority-$\bmp$ matchings and second for Pareto optimal matchings. The latter is an extension of the former and since their proofs are similar, we present only the proof for Pareto optimal matchings in the main body, relegating the one for priority-$\bmp$ matchings to \Cref{app:missingProofsnonAdap}. Informally, the main idea in the algorithm for Pareto optimal matchings is to first carefully choose a set of values $\{t_k\}_ k \in [n]$ and then ask every agent $a_i$ if an object $h_j$ of value at least $t_r$, where $r = rank(a_i, h_j)$. In particular, we set $t_1 =  \frac{1}{n^{1/3}}$ and $t_i = \frac{1}{\min\{i, n^{1/3}\}\cdot{n}^{2/3}}$, for all $i \in \{2,\ldots,n\}$. Next, we draw a bipartite graph where there is an edge between an agent $a_i$ and object $h_j$, if the agent answered ``Yes" w.r.t.\ this object; this edge is assigned a weight $t_r$. Following this, we find the max-weight matching in this graph. Note that not all the agents/objects may have been matched. Therefore, we compute what we refer to as an auxiliary matching, where the auxiliary matching is a matching that maximizes the number of agents who are matched with an edge of rank at most $\floor{\sqrt[3]{n}/2}$ in the case when the max-weight matching in the previous step has no edge of rank-1 (which happens when every agent responds ``NO"  to the query w.r.t.\ their rank-1 object), and is the matching that maximizes the number of agents who are matched with a rank-$1$ edge otherwise. Finally, we combine the max-weight matching computed in the first step with the auxiliary matching (by adding unmatched pairs from the auxiliary one to the max-weight one), arbitrarily match any leftover pairs, and run TTC with the resulting matching as initial endowment. In the theorem below, we show that this achieves an $O(n^{2/3})$-approximation to the optimal social welfare among PO matchings. 
\begin{algorithm}[tb]
	{\small \centering
		\noindent\fbox{%
			\begin{varwidth}{\dimexpr\linewidth-4\fboxsep-4\fboxrule\relax}
				\begin{algorithmic}[1]
					\small 
					\Input an instance $\mi =(\mn, \mh, \mpp = (P_1, \ldots, P_n))$ and priorities $\bmp = (p_1, \ldots, p_n)$
					\Output a priority-$\bmp$ matching when $\bmp \in \mathbb{P}$
					
					\State $\mathbb{G}_{\mi} = (\mn \cup \mh, \me) \gets$ graph induced by $\mi$
					\State $t_1 \gets \frac{1}{n^{1/3}}$
					\State $t_i \gets \frac{1}{\min\{i, n^{1/3}\}\cdot{n}^{2/3}}$, for all $i \in \{2,\ldots,n\}$								
					\For{$e = (a_i, h_j) \in \mn \times \mh$}
					\State $r \gets \text{rank}(a_i, h_j)$
					\If{$\mq(a_i, h_j, t_r)$} 
					\State $w_e \gets p_r + t_r$
					\Else 
					\State $w_e \gets p_r$
					\EndIf
					\EndFor	
					\State $\mu \gets$ max-weight matching in $\mathbb{G}_{\mi}$, where weights are $(w_e)_{e \in \me}$
					\label{algo:line:muMM}		
					\State\Return $\mu$    			
				\end{algorithmic}
		\end{varwidth}}
		\caption{returns a priority-$\bmp$ matching that achieves an $O(n^{2/3})$-approximation to the optimal social welfare among priority-$p$ matchings for the case when the agents have unit-sum valuations.}
		\label{algo:UB-nonAdap-Priority}
	}
\end{algorithm}

\begin{algorithm}[tb]
	{\small \centering
		\noindent\fbox{%
			\begin{varwidth}{\dimexpr\linewidth-4\fboxsep-4\fboxrule\relax}
				\begin{algorithmic}[1]
					\small 
					\Input an instance $\mi = (\mn, \mh, \mpp = (P_1, \ldots, P_n))$
					\Output a Pareto optimal matching
					
					\State $\mu'_{MM} \gets$ matching returned by \Cref{algo:UB-nonAdap-Priority} on $\mi$ and $\bmp = (0, \dots, 0)$		
					\State $\mu_{MM} \gets \mu'_{MM} \setminus \{e \in \mu_{MM} \mid w_e = 0\}$  \Comment{{\footnotesize remove edges with weight 0 from $\mu'_{MM}$}} \label{algo:line:muMM-PO}
					\State $\mu^1_{MM} = \{(a_i, h_j) \mid (a_i, h_j) \in \mu_{MM} \text{ and } rank(a_i, h_j) = 1\}$
					\If{$|\mu^1_{MM}| == 0$}
					\State $\mu'_{aux} \gets$ matching in $\mathbb{G}_{\mi}$ that maximizes the number of agents who are matched with an edge of rank at most $\floor{\sqrt[3]{n}/2}$ \label{line:algo:rankk}					
					\Else
					\State $\mu'_{aux} \gets$ matching in $\mathbb{G}_{\mi}$ where as many agents as possible to a rank-$1$ edge \label{line:algo:rank1}			
					\EndIf
					\State $\mu_{aux} \gets \mu'_{aux} \setminus \{(a, o) \mid (a, o) \in \mu'_{aux} \text{ and either $a$ or $o$ is matched in $\mu_{MM}$}\}$
					
					\State $\mu_{\text{rest}} \gets$ arbitrarily match the acceptable (agent, object) pairs that are not matched in $\mu_{MM} \cup \mu_{aux}$
					\State $\mu \gets $ run $\ttc$ with $\mu_{MM} \cup \mu_{aux} \cup \mu_{\text{rest}}$ as initial endowments and return the resulting matching. 		
					
					\State\Return $\mu$    			
				\end{algorithmic}
		\end{varwidth}}
		\caption{returns a PO matching achieves an $O(n^{2/3})$-approximation to the optimal social welfare among PO matchings for the case when the agents have unit-sum valuations.}
		\label{algo:UB-nonAdap-PO}
	}
\end{algorithm}

\begin{theorem} \label{thm:UB-nonAdap-PO}
	Given an instance $\mi = (\mn, \mh, \mpp = (P_1, \ldots, P_n))$, Algorithm~\ref{algo:UB-nonAdap-PO} asks one non-adaptive query per (agent, object) pair and returns a Pareto optimal matching that achieves an $O(n^{2/3})$-approximation to the optimal welfare among all Pareto optimal matchings for the case when agents have unit-sum valuations. 
\end{theorem}

Before we prove this, we introduce the following notations and terminologies which will be useful. Let $\mu^{*}$ denote the matching that achieves optimal welfare among all Pareto optimal matchings when the agents have unit-sum valuations. Let $H_i$ denote the set of agents who are matched to their $i$-th choice in $\mu^{*}$ and have value at least $t_i$ for their partner in $\mu^{*}$. Similarly, let $L_i$ denote the set of agents who are matched to their $i$-th choice in $\mu^{*}$ and have value less than $t_i$ for their partner in $\mu^{*}$. Additionally, let $H = \cup_{i=1}^n H_i$, $L = \cup_{i=1}^{n} L_i$, and, for some $S \subseteq \mn$, $\mu^{*}_S \subseteq \mu^{*}$ be the set of edges $(a_i, h_j)$ such that $a_i \in S$ and $(a_i, h_j) \in \mu^{*}$. Now, if $\sw(\mu_S)$ denotes the sum of values of the edges in $\mu_S$ (calculated based on the true utilities of the agents), then note that
$\sw({\mu^*}) = \sw(\mu^{*}_H) +  \sw(\mu^{*}_L) = \sum_{(a_i,h_j) \in \mu^{*}_H} v_i(h_j) + \sum_{(a_i,h_j) \in \mu^{*}_L} v_i(h_j)$.

Next, we prove the following lemma.

\begin{lemma} \label{lemma:highVals-PO}
	Let $\mu_{MM}$ be the matching that is computed in line~\ref{algo:line:muMM-PO} in \Cref{algo:UB-nonAdap-PO}. Then, $\sw(\mu^{*}_H) \leq n^{2/3} \cdot \sw(\mu_{MM})$.
\end{lemma}

\begin{proof}
	Note that $\mu^{*}_H \subseteq \me$ and for every $e = (a_i, h_j^\ell) \in \mu^{*}_H$, where $h^\ell_j$ is a $\ell$-th choice of agent $a_i$, the weight $t_\ell$ that is assigned to this edge in $\mathbb{G}_{\mi}$ is at least $\frac{1}{{n}^{2/3}}\cdot v_i(h_j^\ell)$ (and is at most $v_i(h_j^\ell)$). This is so because $t_1 =\frac{1}{n^{1/3}}$, $t_i = \frac{1}{\min\{i, n^{1/3}\}\cdot{n}^{2/3}}$ for $i \geq 2$, and $v_i(h_j^\ell) \leq \frac{1}{\ell}$, since the valuations are unit-sum and $h_j^\ell$ is in the $\ell$-th choice of agent $i$. Combining these two observations, we have that
	$\sw(\mu^{*}_H) = \sum_{(a_i,h_j) \in \mu^{*}_H} v_i(h_j) \leq n^{2/3} \cdot \sw(\mu_{MM})$, 
	where the last inequality follows by using the fact that $\mu_{MM}$ (computed in line~\ref{algo:line:muMM-PO}) is the max-weight matching on the graph $\mathbb{G}_{\mi}$, and as discussed above $\mu^{*}_H \subseteq \me$ and the edge weights in $\mathbb{G}_{\mi}$ are off by a factor of at most $n^{2/3}$.
\end{proof}

Equipped with this lemma, we can now prove our theorem.

\begin{proof}[Proof of \Cref{thm:UB-nonAdap-PO}] 
	First, it is easy to see that \Cref{algo:UB-nonAdap-PO} returns a matching $\mu$ that is Pareto optimal (since it is the matching returned by running the TTC algorithm on some initial endowments). Next, let $\mu_{MM}$ be the matching computed in line~\ref{algo:line:muMM-PO} in \Cref{algo:UB-nonAdap-PO}, and let $\mu_{MM}^1 \subseteq \mu_{MM}$ be set of edges of rank-1 in $\mu_{MM}$. We will proceed by considering the following two cases and show that in each case $\frac{\sw({\mu^*})}{\sw({\mu})} \in O(n^{2/3})$.
	
	\proofcase{Case 1: $|\mu_{MM}^1| \geq 1$.} For this case, let $\mu'_{aux}$ be the matching in $\mathbb{G}_{\mi}$ where as many agents as possible to a rank-$1$ edge (see line~\ref{line:algo:rank1} in \Cref{algo:UB-nonAdap-PO}) and let $\mu_{aux} = \mu'_{aux} \setminus \{(a, o) \mid (a, o) \in \mu'_{aux} \text{ and either $a$ or $o$ is matched in }\allowbreak \mu_{MM}\}$. Below, we will argue that $\sw(\mu_{MM}) + \sw(\mu_{aux}) \geq \max\left\{\frac{|X'_1|}{2n}, t_1\right\}$, where ${X'_i}$ is the set of edges of rank-$i$ in $\mu'_{aux}$.
	
	\begin{claim} \label{clm:clm1}
		If $|\mu^1_{MM}| \geq 1$, then $\sw(\mu_{MM}) + \sw(\mu_{aux}) \geq \max\left\{\frac{|X'_1|}{2n}, t_1\right\}$.
	\end{claim}
	\begin{proof*}
		Consider the set of edges in $X'_1$, and let $X_{aux}$ be the edges in $X'_1$ that belong to $\mu_{aux}$ (i.e., $X_{aux} = X'_1 \cap \mu_{aux}$). By our definition of $\mu_{aux}$, for each $(a_i, h_j) \in X'_1 \setminus X_{aux}$, at least one of $a_i$ or $h_j$ is matched in $\mu_{MM}$. Therefore,
		\begin{align*}
			\sw(\mu_{MM}) + \sw(\mu_{aux}) &\geq \frac{|X'_1 \setminus X_{aux}|}{2}\cdot t_n + \sw({X_{aux}}) \nonumber \\
			&\geq  \frac{|X'_1 \setminus X_{aux}|}{2}\cdot t_n + |X_{aux}|\cdot\frac{1}{n} \nonumber \\
			&\geq \frac{|X'_1|}{2}\cdot\frac{1}{n}. 
		\end{align*}
		In the set of inequalities above, the first inequality follows from the fact that there are at least $\frac{|X'_1\setminus X_{aux}|}{2}$ unique edges in $\mu_{MM}$ and each of them have weight at least $t_n$; the second inequality follows because each edge in $X_{aux}$ is of value at least $\frac{1}{n}$ (since $X_{aux} = X'_1 \cap \mu_{aux}$, the valuations are unit-sum, and the agents are matched to their first choice in $\mu_{aux}$ and so have value at least $\frac{1}{n}$).
		
		Additionally, note that since $|\mu^1_{MM}| \geq 1$, $\sw(\mu_{MM}) \geq t_1$. Hence, combining the two observations above, we have that $\sw(\mu_{MM}) + \sw(\mu_{aux}) \geq \max\{\frac{|X'_1|}{2n}, t_1\}$.	
	\end{proof*}
	
	Equipped with the claim, next, consider $\mu^{*}_L$ and note that,
	\begin{align}
		\sw(\mu^{*}_L) &= \sw(\mu^{*}_{L_{1}}) + \sum_{i=2}^{n} \sw(\mu^{*}_{L_{i}}) \nonumber \\
		&\leq |L_1|\cdot t_1 + (n-|L_1|)\cdot t_{2}\nonumber \\
		&\leq |X'_1|\cdot t_1 + n\cdot t_{2}, \label{eqn:clm2}
	\end{align}
	where the first inequality uses the facts that, by the definition of $L_i$, every agent in $\mu^{*}_{L_i}$ has value less than $t_i$ for their partner and that $t_i\geq t_j$ for $i\leq j$ and $i, j \in [n]$, and the final inequality follows since $\mu'_{aux}$ is a matching in $\mathbb{G}_{\mi}$ where as many agents as possible is matched to a rank-$1$ edge.
	
	Given the above, if $\mu' = \mu_{MM} \cup \mu_{aux} \cup \mu_{rest}$, then we have that,
	\begin{align} \label{eqn:main1}
		\frac{\sw({\mu^*})}{\sw({\mu'})} &= \frac{\sw(\mu^{*}_H) +  \sw(\mu^{*}_L)}{\sw(\mu_{MM}) + \sw(\mu_{aux}) + \sw(\mu_{rest})} \nonumber \\
		&\leq \frac{n^{2/3} \cdot \sw(\mu_{MM}) + |X'_1|\cdot t_1 + n\cdot t_{2}}{\sw(\mu_{MM}) + \sw(\mu_{aux})} 	 \nonumber \\
		&\leq \frac{n^{2/3} \cdot \sw(\mu_{MM})}{\sw(\mu_{MM})} + \frac{|X'_1|\cdot t_1}{|X'_1|/(2n)} + \frac{n\cdot t_{2}}{t_1} \nonumber \ \\
		&\leq n^{2/3} + 2{n}^{2/3} + n^{2/3}. 		
	\end{align}
	In the set of inequalities above, the first one follows from using \Cref{lemma:highVals-PO} and (\ref{eqn:clm2}); the second inequality follows from \Cref{clm:clm1}; the last inequality follows from our choice of $t_1$ and $t_2$.
	
	\proofcase{Case 2: $|\mu^1_{MM}| = 0$.} For this case, let $k = \floor{\sqrt[3]{n}/2}$ and $\mu'_{aux}$ be the matching in $\mathbb{G}_{\mi}$ where as many agents as possible to an edge of rank at most $k$ (see line~\ref{line:algo:rankk} in \Cref{algo:UB-nonAdap-PO}) and let $\mu_{aux} = \mu'_{aux} \setminus \{(a, o) \mid (a, o) \in \mu'_{aux} \text{ and either $a$ or $o$ is matched in }\allowbreak \mu_{MM}\}$. Below, we will argue that $\sw(\mu_{MM}) + \sw(\mu_{aux}) \geq \frac{|X'|}{2n}$, where $X' = \cup_{i=1}^{k} {X'_i}$ and ${X'_i}$ is the set of edges of rank-$i$ in $\mu'_{aux}$.
	
	\begin{claim} \label{clm:clm2}
		If $|\mu^1_{MM}| = 0$, then $\sw(\mu_{MM}) + \sw(\mu_{aux}) \geq \frac{|X'|}{2{n}}$.
	\end{claim}
	\begin{proof*}
		Since $|\mu^1_{MM}| = 0$, every agent values their first choice object at a value less than $t_1$ (since otherwise $|\mu^1_{MM}| \geq 1$). This implies that, since their valuations are unit-sum, for $j \in [k]$, their value for a rank-$j$ object is at least $\frac{1}{2n}$ (see \Cref{clm:unitsumVal} in \Cref{app:addClaims} for a proof). Next, like in the proof of \Cref{clm:clm1}, consider $X_{aux} = \mu_{aux} \cap X'$. Since for every edge $(a_i, h_j) \in X' \setminus X_{aux}$, at least one of $a_i$ or $h_j$ is matched in $\mu_{MM}$, and every edge in $\mu_{MM}$ has weight at least $t_n$, we have, 
		\begin{align*}
			\sw(\mu_{MM}) + \sw(\mu_{aux}) &\geq \frac{|X' \setminus X_{aux}|}{2}\cdot t_n + \sw({X_{aux}}) \nonumber \\
			&\geq  \frac{|X' \setminus X_{aux}|}{2}\cdot t_n + |X_{aux}|\cdot\frac{1}{2n} \nonumber \\
			&\geq \frac{|X'|}{2n}.  \qedhere
		\end{align*}	
	\end{proof*}
	Given the claim, next, note that, 
	\begin{align}
		\sw(\mu^{*}_L) &= \sum_{i=1}^{k}  \sw(\mu^{*}_{L_{i}}) + \sum_{i=k+1}^{n} \sw(\mu^{*}_{L_{i}}) \nonumber \\
		&\leq \left(\sum_{i=1}^{k} |L_i| \right) \cdot t_1 + \left(\sum_{i=k+1}^{n} |L_i| \right)\cdot t_{k+1} \nonumber \\
		&\leq \left(\sum_{i=1}^{k} |X'_i| \right)\cdot t_1 + \left(\sum_{i=k+1}^{n} |L_i| \right)\cdot t_{k+1} \nonumber \\
		&= \abs{X'}\cdot t_1 + \left(\sum_{i=k+1}^{n} |L_i| \right)\cdot t_{k+1}.
		\label{eqn:clm5}
	\end{align}
	
	Equipped with the above, if $\mu' = \mu_{MM} \cup \mu_{aux} \cup \mu_{rest}$, then for this case we have that,
	\begin{align} \label{eqn:main2}
		\frac{\sw({\mu^*})}{\sw({\mu'})} &= \frac{\sw(\mu^{*}_H) +  \sw(\mu^{*}_L)}{\sw(\mu_{MM}) + \sw(\mu_{aux}) + \sw(\mu_{rest})} \nonumber \\
		&\leq \frac{n^{2/3} \cdot \sw(\mu_{MM}) + \abs{X'}\cdot t_1 + \left(\sum_{i=k+1}^{n} |L_i| \right)\cdot t_{k+1}}{\sw(\mu_{MM}) + \sw(\mu_{aux})}  \nonumber \\
		&\leq \frac{n^{2/3} \cdot \sw(\mu_{MM})}{\sw(\mu_{MM})} + \frac{|X'|\cdot t_1}{|X'|/(2n)} + \frac{\left(\sum_{i=k+1}^{n} |L_i| \right)\cdot t_{k+1}}{|X'|/(2n)} \nonumber  \\
		&\leq n^{2/3} + 2{n}^{2/3} + 8{n}^{2/3}. 
	\end{align}
	In the set of inequalities above, the first one follows from using \Cref{lemma:highVals-PO} and (\ref{eqn:clm5}); the second inequality follows \Cref{clm:clm2}; the last inequality follows since $\abs{X'} \geq \min\left\{k, \sum_{i=k+1}^{n} |L_i|\right\}$ (see \Cref{clm:boundX} in \Cref{app:addClaims} for a proof).
	
	Finally, combining (\ref{eqn:main1}) and (\ref{eqn:main2}), and using the fact that the TTC algorithm is individually rational (which in turn implies $SW(\mu) \geq SW(\mu')$), gives us our theorem.
\end{proof}

Next, we state the following result for priority-$\bmp$ matchings, whose proof, as mentioned previously, is similar the proof above and hence appears in \Cref{app:missing:UB-nonAdap-priority}.

\begin{restatable}{theorem}{nonAdapPriority} \label{thm:UB-nonAdap-priority}
	Given an instance $\mi = (\mn, \mh, \mpp = (P_1, \ldots, P_n))$ and a vector of priorities $\bmp = (p_1, \ldots, p_n)$, where $\bmp \in \mathbbm{P}$, Algorithm~\ref{algo:UB-nonAdap-Priority} asks one non-adaptive query per (agent, object)  and returns a priority-$\bmp$ matching that achieves an $O(n^{2/3})$-approximation to the optimal welfare among all priority-$\bmp$ matchings for the case when agents have unit-sum valuations. 
\end{restatable}

Finally, we also consider the case when agents have unit-range valuations and show how one can obtain an $O(\sqrt{n})$-approximation to the optimal social welfare among Pareto optimal and priority-$\bmp$ matchings. Since the algorithms and analyses for these are somewhat similar to the ones above, we present these results in Appendix~\ref{app:sec:unit-range-UB}. 	
	\section{Lower Bounds}
Here we turn our attention to lower bounds for the case when an algorithm can ask at most one query per (agent, object) pair---i.e., for the setting considered in \Cref{sec:nonAdap}. We show that, for the unit-sum and unit-range cases, any deterministic algorithm $\ma$ that asks at most one query per (agent, object) pair and produces a Pareto-optimal/rank-maximal/max-cardinality rank-maximal/fair matching has a worst-case welfare loss of $\Omega(\sqrt{n})$, i.e., $\ml(\ma) \in \Omega(\sqrt{n})$.

\begin{theorem} \label{thm:LB-nonAdap}
	Let $X$ denote one of the properties in the set $\{${Pareto-optimal, $\allowbreak$ rank-maximal, max-cardinality rank-maximal, and fair}$\}$. Let $\ma$ be a non-adaptive deterministic algorithm that always produces a matching that satisfies property $X$ and asks at most one query per (agent, object) pair. If there are $n$ agents with unit-sum valuation functions, then $\ml(\ma) \in \Omega(\sqrt{n})$.
\end{theorem}
\newcommand\perBlock{5}
\newcommand\ntypes{4}
\begin{proof}Let $n \geq 18$ and $n = \perBlock k + r$, where $k = \floor{\frac{n}{5}}$ and $0\leq r \leq \ntypes$. Next, let us construct an instance $\mi$ with the set of agents $\mn$, where $|\mn| = n$ and the set of object $\mh$, where $|\mh| = n$. We partition the set of agents into sets $B_1,\ldots, B_{k+1}$ such that $|B_{k+1}| = r$ and $|B_i |= \perBlock$ for all $i \in [k]$. We will refer to each $B_i$, where $i \in [k]$, as a \textit{block}. Given this, the preferences of the agents are as defined below, where $h_j \succ \mh \setminus \{h_{j}\}$ implies that all the objects in the set $\mh \setminus \{h_{j}\}$ are less preferred than $h_j$ and are preferred in some arbitrary linear order (which is the same for all the agents).
	\begin{align*}
		\forall i \in [k], \text{ agents in $B_i$ have the following preferences}&: h_1 \succ h_{i+1} \succ \mh \setminus \{h_1, h_{i+1}\}\\
		\text{ agents in $B_{k+1}$ have the following preferences}&: h_1 \succ h_{n-1} \succ \mh \setminus \{h_1, h_{n-1}\}
	\end{align*}
	Next, since $\ma$ is a non-adaptive deterministic algorithm that asks at most one query per (agent, object) pair, we can think of $\ma$ to consist of two components---an outcome function $f$ which for an instance $\mi \in \mathbb{I}$ outputs a matching in $\mm_{\bbg_{\mi}}$, and a $|\mn| \times |\mh|$ matrix $T$, where $T_{ij} \in [0, 1]$ represents the threshold asked to agent $a_i$ w.r.t.\ the object $h$ such that $rank(a_i, h) = j$ (i.e., if $\ma$ asks the query $\mq(a_i, h_i^j, t)$, where $rank(a_i, h_i^j) = j$, then $T_{ij} = t$). Note that, for all $i \in [n]$, depending on the values of $T_{i1}$ and $T_{i2}$, we can classify $a_i$ as belonging to one of the following four types.
	\begin{align*}
		&\text{Type-1}: T_{i1} \in [0, \frac{1}{2}), \; T_{i2} \in [0, \frac{1}{\sqrt{n}}]  
		&&\text{Type-2}: T_{i1} \in [\frac{1}{2}, 1], \; T_{i2} \in [0, \frac{1}{\sqrt{n}}]\\ 
		&\text{Type-3}: T_{i1} \in [0, \frac{1}{\sqrt{n}}), \; T_{i2} \in (\frac{1}{\sqrt{n}}, 1] 
		&&\text{Type-4}: T_{i1} \in [\frac{1}{\sqrt{n}}, 1], \; T_{i2} \in (\frac{1}{\sqrt{n}}, 1] 
	\end{align*}
	
	Also, note that, for $i \in [k]$, since each block $B_i$ has $\perBlock$ agents, at least two of them will be classified as the same type (where types are as defined above). For each block $B_i$ and $j \in [\ntypes]$, we say that $B_i$ is a \textit{Type-$j$ block} if there are at least two agents such that both are of Type-$j$ and for every $1\leq k < j$, there is at most one agent of Type-$k$. We use $n_j$ to denote the number of Type-$j$ blocks, and hence have that $\sum_{j\in[4]} n_j = k$. Additionally, for a Type-$j$ block $B_i$, we let $a_{2i-1}$ and $a_{2i}$ to denote the two agents who are of Type-$j$ in $B_i$ (if there are more than two such agents, consider any two of them and denote them as $a_{2i-1}$ and $a_{2i}$) and say that $a_{2i-1}$ and $a_{2i}$ are the \textit{special agents} in $B_i$.
	
	As a final step before proving the theorem, we need to define a set of valuation profiles that are consistent with the preferences mentioned above. Our objective will be to show that there is at least one valuation profile in this set which achieves the desired bound on $\ml(\ma)$. So, to do this, let us first define the following unit-sum utility functions,\footnote{Note that in each of these valuations, objects 3 to $n$ have been given the same value. This is purely for ease of exposition. To ensure that no two objects have the same value, we can just perturb the values slightly.}
	\begin{align*}
		&u_0 = (c_1^2 + \epsilon/2, c_1^2, c_1^2 - c_2 \epsilon,\ldots, c_1^2 - c_2 \epsilon)\\	 
		&u_1 = (1-c_1, c_1, 0, \ldots, 0) &&u_2 = (1/2+\epsilon, 1/2-\epsilon, 0, \ldots, 0) \\
		&u_3 = (1/2 - c_1 - \epsilon, c_1 + \epsilon, c_2, \ldots, c_2) &&u_4 = (1/4+\epsilon, 1/4-\epsilon, c_2, \ldots, c_2)\\
		&u_5 = (1-c_1^2, c_1^2, 0 \ldots, 0) &&u_6 = (1-c_1 + \epsilon, c_1 - \epsilon, 0, \ldots, 0)\\ 
		&u_7 = (c_1-c_1^2, c_1^2, c_3, \ldots, c_3) &&u_8 = (3c_1/4, c_1/4, c_3, \ldots, c_3),
	\end{align*} 
	where $u_j = (x_1,x_2, x_3\ldots, x_3)$ implies that they value their most preferred object at $x_1$, second most preferred object at $x_2$, and every other object at $x_3$, $\epsilon > 0$ is a very small real number ($\epsilon < 1/n^4$ will suffice), $c_1 = \frac{1}{\sqrt{n}}$, $c_2 = \frac{1}{2(n-2)}$, and $c_3 = \frac{1-c_1}{n-2}$. Given this, consider the following set of valuation profiles $\mathcal{V}_\mi$, all of which are consistent with the preferences in $\mi$ and where each of them is defined the following way.
	\begin{enumerate}[label=\roman*),ref=\roman*]
		\item for $i \in [k]$, $j\in [\ntypes]$, and special agents $a_{2i-1}$ and $a_{2i}$ in $B_i$, where $B_i$ is a Type-$j$ block, let one of the agents have utility function $u_{2j-1}$ and the other have $u_{2j}$.
		\item let the agents in $\mn \setminus \{a_1,\ldots,a_{2k}\}$ have utility function $u_0$.
	\end{enumerate} 
	Note that the only difference between valuation profiles in $\mathcal{V}_\mi$ is w.r.t.\ the utility functions of the special agents. Additionally, note that for each $v \in \mathcal{V}_{\mi}$ and for each Type-$j$ block $B_i$, the utility functions, $u_{2j-1}$ and $u_{2j}$, of the special agents have been defined in such a way that they will respond identically to the queries in $T$. This can be seen by considering \Cref{tab:response} which shows how for all $j \in [4]$, the special agents of a Type-$j$ block will respond to queries $\mq(\cdot , \cdot, T_{i1})$ or $\mq(\cdot , \cdot, T_{i2})$, and by observing that utility functions $u_{2j-1}$ and $u_{2j}$ have the same values for all objects except the two most preferred ones. 
	
	\begin{table}[t]
		\centering 
		{\small
			\begin{tabular}{c|c|c|}
				\cline{2-3}
				& Response to $\mq(\cdot , \cdot, T_{i1})$ & Response to $\mq(\cdot , \cdot, T_{i2})$ \\ \hline
				\multicolumn{1}{|c|}{\begin{tabular}[c]{@{}c@{}}special agents in a Type-1 block\\ (assigned $u_{1}$ and $u_{2}$)\end{tabular}} & Yes                                      & Yes                                      \\ \hline
				\multicolumn{1}{|c|}{\begin{tabular}[c]{@{}c@{}}special agents in a Type-2 block\\ (assigned $u_{3}$ and $u_{4}$)\end{tabular}} & No                                       & Yes                                      \\ \hline
				\multicolumn{1}{|c|}{\begin{tabular}[c]{@{}c@{}}special agents in a Type-3 block\\ (assigned $u_{5}$ and $u_{6}$)\end{tabular}} & Yes                                      & No                                       \\ \hline
				\multicolumn{1}{|c|}{\begin{tabular}[c]{@{}c@{}}special agents in a Type-4 block\\ (assigned $u_{7}$ and $u_{8}$)\end{tabular}}  & No                                       & No                                       \\ \hline
			\end{tabular}
		}
		\caption{Responses of the special agents in a Type-$j$ block to queries $\mq(\cdot , \cdot, T_{i1})$ or $\mq(\cdot , \cdot, T_{i2})$}
		\label{tab:response}
	\end{table}
	
	Now, equipped with all of the above, let us first argue about the welfare-optimal matching that satisfies property $X$. To do this, consider any $v \in \mathcal{V}_{\mi}$ and consider the matching $\mu^*$ of size $n$ where in each Type-$j$ block $B_i$, the special agent with utility function $u_{2j}$ gets $h_{i+1}$ (so they all get their second most preferred object), if $|B_{k+1}| > 0$, then one of the agents in $B_{k+1}$ gets $h_{n-1}$, and the rest of the objects are allocated arbitrarily. From the way the preferences of the agents are defined, it is easy to verify that such a matching satisfies property $X$. Also, if $Z = n-2\sum_{j \in [\ntypes]}n_j$, then from the definition of utility functions $u_{2j}$ for all $j \in [4]$, and since the utility of matching $h_1$ and all other objects to some agent is at least $(c_1^2 + \epsilon/2)$ and $(c_1^2 - c_2\epsilon)$, respectively, we have that, 
	\begin{align}
		\sw(\mu^*) &\geq (c_1^2 + \epsilon/2) +  (1/2-\epsilon)n_1 + (1/4-\epsilon)n_2 + ({c_1-\epsilon})n_3 + \frac{c_1}{4}n_4 + (c_1^2 - c_2\epsilon) Z \label{eqn:lb1} \\
		&\geq (c_1^2 + \epsilon/2) + (n_1 + n_2 + n_3 + n_4)\frac{c_1}{4} + (c_1^2 - c_2\epsilon) Z \nonumber \\
		&\geq \frac{\sqrt{n}}{28}, \label{eqn:lb2}
	\end{align}
	where the last inequality follows since $n = \perBlock k+r$, where $0\leq r \leq \ntypes$, implies $k = \sum_{j \in [\ntypes]}n_j = \floor{n/5}$, and $n\geq 18$.
	
	Finally, consider $\ma$ and let $\mu$ be the matching returned by $\ma$ for the instance $\mi$ when agents are asked queries in $T$. Next, we adversarially pick a $v \in \mathcal{V}_{\mi}$ the following way:
	\begin{enumerate}[label=\roman*),ref=\roman*]
		\item for $i \in [k]$, if $B_i$ is a Type-$j$ block for some $j \in [4]$ and if $\mu(a_{2i-1}) = h_{i+1}$, then consider valuation profile in $\mathcal{V}_{\mi}$ where $a_{2i-1}$ has utility function $u_{2j-1}$ 
		\item for $i \in [k]$, if $B_i$ is a Type-$j$ block for some $j \in [4]$ and if $\mu(a_{2i}) = h_{i+1}$, then consider valuation profile in $\mathcal{V}_{\mi}$ where $a_{2i}$ has utility function $u_{2j-1}$ 
	\end{enumerate} 
	Given such a $v \in \mathcal{V}_{\mi}$, we can now calculate $\sw(\mu)$. And for this, observe that from the definition of utility functions $u_{2j-1}$ for $j \in [4]$, and since the utility of matching $h_1$ and all other objects to some agent is at most $1$ and $c_1^2$, respectively, we have that,
	\begin{align} \label{eqn:lb3}
		\sw(\mu) &\leq 1 + c_1n_1 + (c_1 +\epsilon)n_2 + c_1^2n_3 + {c_1^2}n_4 + c_1^2Z \nonumber\\
		&\leq 2 + \frac{8}{\sqrt{n}}\left(\frac{1}{2}n_1 + (1/4-\epsilon)n_2 + \frac{c_1}{4}n_3 + ({c_1-\epsilon})n_4\right) &&\text{{\small(since $Z \leq n$)}}	  \nonumber \\
		&\leq \frac{56}{\sqrt{n}}\sw(\mu^*) + \frac{8}{\sqrt{n}}\sw(\mu^*). &&\text{{\small(using (\ref{eqn:lb1}) and (\ref{eqn:lb2}))}}
	\end{align} 
	
	Rearranging (\ref{eqn:lb3}) we have that $\frac{\sw(\mu^*)}{\sw(\mu)} \geq \frac{\sqrt{n}}{64}$, or in other words that, $\ml(\ma) \in \Omega(\sqrt{n})$.
\end{proof}

Finally, we also consider the unit-range case and show that any deterministic algorithm that asks at most one query per (agent, object) pair and produces a Pareto-optimal/rank-maximal/max-cardinality rank-maximal/fair matching has $\ml(\ma) \in \Omega(\sqrt{n})$. The proof of this appears in~\ref{app:sec:unit-range-LB} and is almost identical to the proof of \Cref{thm:LB-nonAdap}, with the main difference being in the way utility functions $u_0,u_1,\ldots,u_8$ in the proof are defined.	
	\section{Discussion}
The focus of this paper was on one-sided matching problems. While the usual assumption in such problems is that agents only submit ordinal preferences, it is not hard to imagine scenarios where agents might have some cardinal preferences which, for instance, indicate that they like object $h_1$ much more than object $h_2$. Although ignoring this information can lead to a lose in welfare, asking the agents for their cardinal utilities is not ideal, since determining their exact utilities can be a cognitively-involved task. Therefore, in this paper we investigated the benefit  of eliciting a small amount of extra information about agents' cardinal utilities. In particular, we designed algorithms that used simple threshold queries and returned a matching satisfying some desirable matching property, while also achieving a good approximation to the optimal welfare among all matchings satisfying that property.  Overall, our results show how even asking  agents for just one bit of extra information per object can improve welfare. 

There are a number of future research directions that this work can take. First, the model in this paper assumes that each agent needs to be matched to at most one object and that each object can be matched to at most one agent. However, there are several situations where more than one agent can be matched to the same object, like when assigning students to courses or schools. While our results do not directly hold when each object $h_j$ has a capacity constraint $c_j$, only minimal modifications are needed. In particular, every time we construct a graph in any of the algorithms, all that needs to be done is to create $c_j$ copies for the node that corresponds to object $h_j$. Other open algorithmic problems include addressing the gap between the upper and lower bounds for the non-adaptive algorithms,  expanding the set of properties of interest to include, for example, popular matchings~\cite{abra07}, or asking similar questions in the context of two-sided matching problems. 

More broadly, a particularly interesting direction is to better understand the implications of deploying such an approach in practice.
As mentioned in the introduction, we believe that in many settings non-adaptive algorithms that only ask the agents for a few number of queries with respect to an object might be the most practical approach to pursue, since it involves minimal communication overhead. Moreover, deploying something like that seems easier since the only change that needs to be made to the existing system which asks for ordinal preferences is to add checkboxes with respect to an object and the corresponding threshold queries. Nevertheless, there are still challenges to make this truly useful. A careful reader would have noticed that the thresholds used in our algorithms are very specific values (like $1/n^{1/3}$), which may not be easy to answer. While one potential way to mitigate some of this difficulty is by multiplying all the threshold values by a large enough constant so as to make them easier to comprehend, it's not clear if that would be enough. Therefore, it might be useful to have studies to better understand the kinds of queries that are easier to answer and the types of interface-design that can best support queries, as well as better understand what matching properties are deemed to be most important by users and designers of systems.

\textbf{Acknowledgments. } We thank Hong Zhou for useful discussions.

\printbibliography

\appendix

\section{Additional Claims} \label{app:addClaims}

\begin{claim} \label{clm:unitsumVal}
	For $i \in [n]$, let $a_i \in \mn$ and let $v_i$ be the valuation function of $a_i$, where $v_i$ is a unit-sum valuation function and for all $h \in A_i$ such that $rank(a_i, h) = 1$, $v_i(h) < \frac{1}{n^{1/3}}$. Then, for all $j \in [k]$, where $k = \floor{\sqrt[3]{n}/2}$, and $h \in A_i$ such that $rank(a_i, h) = j$, $v_i(h) \geq \frac{1}{2n}$ 
\end{claim}
\begin{proof}
	Let $t_1 = \frac{1}{n^{1/3}}$. For any $j \in [k]$ and $h \in A_i$ such that $rank(a_i, h) = j$, we will show that $v_i(h) \geq \frac{1-(j-1)t_1}{n-(j-1)}$. To see this, consider the smallest ${j}^*$ such that for $h^* \in A_i$ with $rank(a_i, h^*) = {j}^*$ and $v_i(h^*) < \frac{1-({j}^*-1)t_1}{n-({j}^*-1)}$. Note that ${j}^* \geq 2$, since $v_i$ is a unit-sum valuation function and hence $v_i(h) \geq \frac{1}{n}$ for all $h \in A_i$ such that $rank(a_i, h) = 1$. 
	
	Next, note that, for any $\ell \in \{2, \ldots, n\}$ and $U_{\ell} = \{h \in A_i \mid rank(a_i, h) \leq \ell-1\}$, if there exists an $h \in A_i$ such that $rank(a_i, h) = \ell$, then it follows from the definition of rank that $|U_{\ell}| = \ell-1$ . Therefore, using this below, we have that, 
	\begin{align*}
		\sum_{h\in A_i} {v_i(h)} &= v_i(h^*) + \sum_{\substack{h \in U_{j^*}}}v_i(h) + \sum_{A_i\setminus (U_{j^*} \cup \{h^*\})}v_i(h) \\
		&\leq v_i(h^*) + (j^*-1)\cdot t_1 + (n-j^*)\cdot v_i(h^*)\\
		&<(j^*-1)\cdot t_1 + (n-(j^*-1)) \cdot  \frac{1-(j^*-1)t_1}{n-(j^*-1)}\\
		&=1.
	\end{align*}
	However, this in turn contradicts the fact that $v_i$ is a unit-sum valuation function. 
	
	Therefore, since $t_1 \leq \frac{1}{2k}$, for all $j \in [k]$ and $h \in A_i$ such that $rank(a_i, h) = j$, we have that \[v_i(h) \geq \frac{1-(j-1)t_1}{n-(j-1)} \geq \frac{1-k\cdot t_1}{n-(j-1)} \geq \frac{1-k\cdot \frac{1}{2k}}{n-(j-1)} \geq \frac{1}{2n}.\]	
\end{proof}

\begin{claim} \label{clm:sigTopriority}
	Given an instance $\mi = (\mn, \mh, \mpp = (P_1, \allowbreak \ldots, P_n))$, a matching $\mu$ is a
	\begin{enumerate}[label=\roman*),ref=\roman*]
		\item rank-maximal matching w.r.t.\ $\mi$ if and only if it is a priority-$\bmp$ matching, where $\bmp = (p_1, \allowbreak \ldots, p_n)$ and, for $j \in [n]$, $p_j= n^{2(n-j+1)}$.
		\item max-cardinality rank-maximal matching w.r.t.\ $\mi$ if and only if it is a priority-$\bmp$ matching, where $\bmp = (p_1, \ldots, p_n)$ and, for $j \in [n]$, $p_j= n^{2n} + n^{2(n-j)}$.
		\item fair matching w.r.t.\ $\mi$ if and only if it is a priority-$\bmp$ matching, where $\bmp = (p_1, \ldots, p_n)$ and, for $j \in [n]$, $p_j= 4n^{2n} - 2n^{(j-1)}$.
	\end{enumerate}
\end{claim}
\begin{proof}
	\proofcase{i).} ($\Rightarrow$) Suppose $\mu$ is a priority-$\bmp$ matching (which in turn implies that it is the max-weight matching in $\mathbb{G}_{\mi}$ with weights $\{w_e\}_{e \in \me}$, where, for an edge $e = (a_i, h_j)$ and $r = rank(a_i, h_j)$, $w_e = p_r$), but is not rank-maximal w.r.t.\ $\mi$. Let $\mu'$ be a rank-maximal matching w.r.t.\ $\mi$, and the let the signatures of $\mu$ and $\mu'$ be $(s_1, \ldots, s_n)$ and $(s'_1, \ldots, s'_n)$, respectively. Since $\mu$ is not rank-maximal, consider the smallest $j \in [n]$ such that $s_i = s'_i$ for all $i < j$ and $s_j < s'_j$.	Next, note that both $\mu$ and $\mu'$ are matchings in $\mathbb{G}_{\mi}$, and that the weight of a $j$-th ranked edge is  $n^{2(n-j+1)}$. This in turn implies that weight of $\mu'$ is greater than weight of $\mu$, since $n^{2(n-j+1)} > n \cdot n^{2(n-j)}$, or in other words that taking a $j$-th ranked edge in $\mathbb{G}_{\mi}$ with weights $\{w_e\}_{e \in \me}$ is more beneficial than taking any number of edges of lower rank. However, this contradicts the fact that $\mu$ is the max-weight matching in $\mathbb{G}_{\mi}$.
	
	($\Leftarrow$) Suppose $\mu$ is a rank-maximal matching, but is not a priority-$\bmp$ matching. This implies there is another matching $\mu'$ such that $\mu'$ has a higher weight than $\mu$ in $\mathbb{G}_{\mi}$ with weights $\{w_e\}_{e \in \me}$, where, for an edge $e = (a_i, h_j)$ and $r = rank(a_i, h_j)$, $w_e = p_r$. Also note that $\mu'$ cannot be rank-maximal, since all rank-maximal matchings will have the same weight in $\mathbb{G}_{\mi}$ with weights $\{w_e\}_{e \in \me}$. Given this, it is easy to see that we have a contradiction given our choice of $p_j$s and the fact that $\mu$ is rank-maximal.	
	
	\proofcase{ii).} ($\Rightarrow$) Suppose $\mu$ is a priority-$\bmp$ matching. First we will argue that $\mu$ is a max-cardinality matching in $\mathbb{G}_{\mi}$. To see this, suppose not. This implies there is another matching $\mu'$ such that $|\mu'| > |\mu|$. Next, note that if $W[\mu]$ denotes the weight of $\mu$ in $\mathbb{G}_{\mi}$ with weights $\{w_e\}_{e \in \me}$, then $W[\mu] \leq |\mu| \cdot p_1$. Also, note that we have, $W[\mu'] \geq |\mu'|\cdot n^{2n} > |\mu| \cdot p_1 \geq W[\mu]$, where the second inequality follows from the fact that $|\mu'| - |\mu| \geq 1$ and $|\mu| \leq n$. However, this in turn contradicts the fact that $\mu$ is a priority-$\bmp$ matching.
	
	Now that we have the established that $\mu$ is a max-cardinality matching, we need to show that it is rank-maximal (among max-cardinality matchings). And the proof of this can be obtained by proceeding as in the corresponding case in i) above.
	
	($\Leftarrow$) The proof of this can be obtained by proceeding as in the corresponding case in i) above. 
	
	\proofcase{iii).} This case can be handled as in ii) above, with the only difference being that the signature in this case is the $(n+1)$-tuple $(\sum_{i=1}^n s_i, -s_n, -s_{n-1}, \ldots, -s_1)$.		
\end{proof}

\begin{claim} \label{clm:boundB}
	Given an instance $\mi = (\mn, \mh, \mpp)$ and priorities $\bmp = (p_1, \ldots, p_n)$, where $\bmp \in \mathbb{P}$, let $\mu$ denote a Pareto optimal matching or a priority-$\bmp$ matching w.r.t.\ $\mi$. If $B_i$ denotes the set of agents matched to a rank-$i$ edge in $\mu$, $k$ is a positive integer that is at most $\floor{n/2}$, and $B = \cup_{i=1}^{k} B_i$, then $|B| \geq \min\left\{k, |\mu|\right\}$.  
\end{claim}
\begin{proof}
	Let $B' = \cup_{i=k+1}^{n} B_i$ be the set of agents who are matched to an object of rank at least $k+1$ in $\mu$. In order to prove our claim, let us first consider the case when $k \geq |\mu|$. For this case we will show that $|B'| = 0$, thus implying that $|B| = \mu$. To see this, suppose $|B'| \geq 1$. Without loss of generality, let $a_1 \in |B'|$ and let $O_1^k = \{h \in A_1 \mid rank(a_1, h) \leq k\}$. Since $\mu$ is a Pareto optimal or priority-$\bmp$ matching, all the objects in $O_1^k$ are matched to some agent in $\mn\setminus\{a_1\}$, because if there exists an unmatched, say, $h_1 \in O_1^k$, then $\mu \setminus \{a_1, \mu(a_1)\} \cup \{a_1, h_1\}$ Pareto dominates $\mu$ and has the same size as $\mu$. However, this in turn implies $|\mu| \geq (|O_1^k| + 1) = k+1$, a contradiction. 
	
	Next, let us consider the case when $k < |\mu|$. To prove our claim, let us assume for the sake of contradiction that $|B| < k$, which in turn, along with the fact that $k < |\mu|$, implies that $|B'| \geq 1$. Let $Y = |B'|$ and w.l.o.g.\ let us assume that $\{a_1, \dots a_{Y}\}$ are the agents in $B'$. Note that, for $i \in [Y]$, $a_i$ finds at least $k+1$ objects acceptable since they are matched to an object of rank at least $k+1$. 
	
	Now, consider any $i \in [Y]$ and  let $O_i^k = \{h \in A_i \mid rank(a_i, h) \leq k\}$. As argued above for the case when $k \geq |\mu|$, since $\mu$ is a Pareto optimal matching or a priority-$\bmp$ matching, all the objects in $O_i^k$ are matched to some agent. Additionally, also note that, when combined with the previous observation, there is at least one object, say, $h^s_i \in O_i^k$ such that it is matched to an agent, say, $a_p \in B'$ such that $rank(a_p, h_i^s) \geq k+1$ (because otherwise, $|B| \geq k$, a contradiction); we will refer to $h^s_i$ as $a_i$'s special object. Given all the observations above, we will now argue that if $\mu$ is such that $|B| < k$, then there exists a matching $\mu'$ such that $\mu'$ Pareto dominates $\mu$ and has the same size as $\mu$.
	
	To see this, consider the graph $\mathbb{G'} = (\{a_1, \cdots, a_{Y}\}, \me')$, where there is a directed edge from $a_i$ to $a_j$ if $\mu(a_j) = h^s_i$ (i.e., if $a_j$ is matched to $a_i$'s special object). Since, for every $i \in [Y]$, $a_i$ has a special object $h_i^s$ and $h_i^s$ is matched to one of the agents in $\{a_1, \cdots, a_{Y}\}\setminus\{a_i\}$, note that this graph has to have a cycle. So, now, consider this cycle and implement the trade indicated by this cycle---meaning, if $(a_i, a_j)$ is an edge in this cycle, then allocate $\mu(a_j)$ to $a_i$. Note that the resulting matching, say, $\mu'$ has the same size as $\mu$ and also Pareto dominates $\mu$ (since every agent $a_i$ in this cycle gets an object that is in $O_i^k$), which in turn contradicts the fact that $\mu$ is a Pareto optimal or priority-$\bmp$ matching. 
\end{proof}

\begin{claim} \label{clm:boundX}
	Given an instance $\mi = (\mn, \mh, \mpp)$, let $\mu_1$ be any arbitrary matching in $\mathbb{G}_{\mi}$, and $\mu_2$ denote the matching in $\mathbb{G}_{\mi}$ that matches as many agents as possible with an edge of rank at most $k$, where $k = \floor{\sqrt[3]{n}/2}$. If $L_i$ and $X'_i$ denote the set of agents matched to a rank-$i$ edge in $\mu_1$ and $\mu_2$, respectively, and $X' = \cup_{i=1}^{k} X'_i $, then $|X'| \geq \min\left\{k, \left(\sum_{i=k+1}^{n} |L_i| \right)\right\}$. 
\end{claim}
\begin{proof}
	Let $Y = \sum_{i=k+1}^{n} |L_i|$, i.e., the number of agents who are matched to an object of rank at least $k+1$ in $\mu_1$. Since all the agents in $\cup_{i=k+1}^{n} L_i$ find at least $k+1$ objects acceptable, this implies that there is a matching of size $z = \min\{k, Y\}$ in $\mathbb{G}_{\mi}$ where all the agents are matched to a rank-$i$ edge, where $i \in [k]$. Now, since $\mu_2$ is the matching in $\mathbb{G}_{\mi}$ that matches as many agents as possible with an edge of rank at most $k$, $X' \geq z$.     
\end{proof}

\section{Missing Proofs from \texorpdfstring{\Cref{sec:nonAdap}}{}} \label{app:missingProofsnonAdap}

\subsection{Proof of \texorpdfstring{\Cref{thm:UB-nonAdap-priority}}{}} \label{app:missing:UB-nonAdap-priority}
\nonAdapPriority*

To prove this, we use the same notations and terminologies that were introduced for the proof of \Cref{thm:UB-nonAdap-PO}, except that now these are defined with respect to priority-$\bmp$ matchings. Next, we prove the following lemma, which is almost identical to \Cref{lemma:highVals-PO}.
%

\begin{lemma} \label{lemma:highVals-Priority}
	Let $\mu$ be the matching that is computed in line~\ref{algo:line:muMM} in \Cref{algo:UB-nonAdap-Priority}. Then, $\sw(\mu^{*}_H) \leq n^{2/3} \cdot \sw(\mu)$.
\end{lemma}

\begin{proof}
	This can be proved by proceeding exactly as in the proof of \Cref{lemma:highVals-PO}.
\end{proof}


\begin{proof}[Proof of \Cref{thm:UB-nonAdap-priority}] First, it is easy to see that for the given priority vector $\bmp \in \mathbb{P}$ the matching returned by \Cref{algo:UB-nonAdap-Priority} is a priority-$\bmp$ matching. Next, let $S$ be the set of agents who answered ``Yes" w.r.t.\ $t_1$ (i.e., all these agent have a value of at least $t_1$ for (one of) their first choice object(s)).	Also, let $B_i$ be the number of agents who are matched to (one of) their $i$-th choice object(s) in $\mu$. Note that since $\mu$ is a priority-$\bmp$ matching, we know that $|B_i| = |H_i| + |L_i|$ for all $i \in [n]$ (since all priority-$\bmp$ matchings have the same signature). Additionally, we also know that $\sw(\mu) \geq |B_1| \cdot \frac{1}{n}$, since the agents have unit-sum valuations.	
	
	Now, if $|S| \geq 1 $, then we have that, 
	\begin{align} \label{lemm:p-eqn1}
		\frac{\sw({\mu^*})}{\sw({\mu})} &= \frac{\sw(\mu^{*}_H) +  \sw(\mu^{*}_L)}{\sw(\mu)} \nonumber \\
		&= \frac{\sw(\mu^{*}_H) +  \sw(\mu^{*}_{L_{1}}) + \sum_{i=2}^{n} \sw(\mu^{*}_{L_{i}})}{\sw(\mu)} \nonumber\\
		&\leq \frac{n^{2/3} \cdot \sw(\mu) + |B_1|\cdot t_1 + n\cdot t_{2}}{\sw(\mu) } \nonumber\\ 
		&\leq \frac{n^{2/3} \cdot \sw(\mu)}{\sw(\mu)} + \frac{|B_1|\cdot t_1}{|B_1|\cdot \frac{1}{n}} + \frac{n\cdot t_{2}}{t_1} \nonumber\\ 
		&\leq n^{2/3} + {n}^{2/3} + n^{2/3}, 
	\end{align}
	where the first inequality by using \Cref{lemma:highVals-Priority} and the fact that $L_1 \leq B_1$, second inequality follows from the fact that $\sw(\mu) \geq t_1$, since $|S| \geq 1 $. and the final inequality follows from our choice of $t_1$ and $t_2$. 
	
	On the other hand, if $|S| = 0$, then, first, let $k = \floor{\sqrt[3]{n}/2}$. Next, note that every agent values their first choice at a value less than $t_1$. This in turn implies that, since their valuations are unit-sum, for $j \in [k]$, their value for a rank-$j$ object is at least $\frac{1}{2n}$ (see \Cref{clm:unitsumVal} in \Cref{app:addClaims} for a proof). Additionally, let $B = \cup_{i=1}^{k} {B_i}$ and $B' = \cup_{i=k+1}^{n} {B_i}$; from \Cref{clm:boundB} we know that $|B| \geq \min\{k, |\mu|\}$. Given this, we have, 
	\begin{align} \label{lemm:p-eqn2}
		\frac{\sw({\mu^*})}{\sw({\mu})} &= \frac{\sw(\mu^{*}_H) +  \sw(\mu^{*}_L)}{\sw(\mu)} \\
		&=\frac{\sw(\mu^{*}_H) +  \sum_{i=1}^{k}  \sw(\mu^{*}_{L_{i}}) + \sum_{i=k+1}^{n} \sw(\mu^{*}_{L_{i}})}{\sw(\mu)} \nonumber\\
		&\leq \frac{\sw(\mu^{*}_H) + \left(\sum_{i=1}^{k} |B_i| \right)\cdot t_1 + \left(\sum_{i=k+1}^{n} |B_i| \right)\cdot t_{k+1}}{\sw(\mu)} \nonumber \\ 
		&=\frac{n^{2/3} \cdot \sw(\mu) + |B|\cdot t_1 + |B'|\cdot t_{k+1}}{\sw(\mu)} \nonumber \\
		&\leq \frac{n^{2/3} \cdot \sw(\mu)}{\sw(\mu)} + \frac{|B|\cdot t_1}{|B|/(2n)} + \frac{|B'|\cdot t_{k+1}}{|B|/(2n)} \nonumber \\ 
		&\leq n^{2/3} + 2{n}^{2/3} + 8{n}^{2/3}, 
	\end{align}
	where the first inequality follows from using \Cref{lemma:highVals-Priority} and the fact that $L_i \leq B_i$ for all $i \in [n]$, the second inequality follows from \Cref{clm:clm2}, and the final inequality follows since $|B| \geq \min\{k, |\mu|\}$ and $t_{{k}+1} \leq 2/n$.
	
	Finally, combining (\ref{lemm:p-eqn1}) and (\ref{lemm:p-eqn2}) gives us our theorem.
\end{proof}

\section{Additional Discussions} 

\subsection{Power of ordinal algorithms} \label{app:sec:ordinal}

In this section we look at the power of ordinal algorithms---meaning we want to understand the worst-case loss in welfare experienced by an ordinal algorithm (i.e., an algorithm which only uses the ordinal information given by the agents). We argue that when agents have unit-sum valuations functions, any deterministic algorithm $\ma$ is such that $\ml(\ma) \in \Omega({n}^2)$. The proof of this is similar to the proof of result by \citeauthor{aman20b} where they show that the distortion of any deterministic ordinal algorithm is $\Omega({n}^2)$ \cite[Thm.\ 1]{aman20b}.  Also, just like in \citet[Thm.\ 1]{aman20b}, this bound is asymptotically tight.

\LBordinal*
\begin{proof}[Proof (sketch)]
	The proof here follows almost directly from the proof of Theorem 1 in the paper by \citet{aman20b} where they show that the distortion of any deterministic ordinal algorithm is $\Omega({n}^2)$. 
	
	The main observation to note is that for the instance they construct if there are $j$ pairs of agents and if an agent in the $i$-th pair gets the top-choice object $a$, then any matching $\mm$ that matches the other agent in the $i$-th pair to their second choice $b_i$, one of the agents in pair $\ell \in [j]\setminus \{i\}$ to their second choice $b_{\ell}$, and the remaining agents to their highest possible choice (i.e., for pair $\ell \in [j]\setminus \{i\}$, allocate object $c_{\ell-1}$ to the agent who has not received their second choice) is Pareto-optimal/rank-maximal/max-cardinality rank-maximal/fair. 
	
	Given this, one can construct valuation functions as described in their proof and see that an ordinal algorithm cannot distinguish between matchings that have welfare of at least $n/4$ and ones which have a welfare of at most $1/n$, thus resulting in a lower bound of $\Omega(n^2)$.
	
	Finally, to see that the bound is tight, first note that when agents have unit-sum valuation functions, as long as there is at least one agent who is matched to their top-choice object, $\ml(\ma) \in O({n}^2)$, since the maximum social welfare achievable is $n$, and the agent who gets their top-choice has a value of at least $1/n$ for their top-choice object. Now, in order to achieve this, one can run any Pareto-optimal/rank-maximal/max-cardinality rank-maximal/fair algorithm and the matching returned by any such algorithm has at least one agent matched to a rank-1 object (see \Cref{clm:boundB} in \Cref{app:addClaims} for a proof).
\end{proof}

One can derive a similar lower bound for unit-range valuations and show that it is $\Omega(n)$. This can be done by making similar observations as in the proof of \Cref{thm:LB-ordinal}. The only part that needs to be modified is in the way the valuation functions are defined; all we need to do here is to define the value of each top-choice object to be 1 and the least preferred object to be 0.  Moreover, this is again asymptotically tight because of the exact same reason mentioned above for unit-sum valuations. The only difference is that when agent have unit-range valuations, the value of the top-choice object is 1, thus resulting in an $O(n)$ bound. 

\begin{theorem} \label{thm:LB-ordinal-unitRange}
	Let $X$ denote one of the properties in the set $\{${Pareto-optimal, $\allowbreak$ rank-maximal, max-cardinality rank-maximal, and fair}$\}$. Let $\ma$ be a deterministic ordinal algorithm that always produces a matching with property $X$. If there are $n$ agents with unit-range valuation functions, then $\ml(\ma) \in \Omega({n})$. Moreover, this bound is asymptotically tight.
\end{theorem}

\subsection{The unit-range case} \label{app:sec:unit-range}
In this section we discuss the case when agents have unit-range valuations. Note that the adaptive algorithm presented in \Cref{sec:Adap} works for both the unit-sum and unit-range case. Therefore, here we look at the non-adaptive case, in particular focussing on the case when an algorithm is allowed to ask at most one query per (agent, object) pair.

\subsubsection{Improving welfare when asking one query per (agent, object) pair} \label{app:sec:unit-range-UB}
Here we show an algorithm that achieves an $O(\sqrt{n})$-approximation to the optimal welfare among all Pareto optimal or priority-$\bmp$ matchings when it is allowed to ask at most one query per (agent, object) pair. In the next section, we will show that this is asymptotically optimal. At a high-level, the algorithm is very similar to \Cref{algo:UB-nonAdap-Priority}, with the main difference being that instead of $t_i$s defined there, where we used $t_1 = \frac{1}{n^{1/3}}$ and $t_i = \frac{1}{\min\{i, n^{1/3}\}\cdot{n}^{2/3}}$ for all $i \in \{2,\ldots,n\}$, here the values are more uniform and we use $t_1 = 1$ and $t_i = \frac{1}{\sqrt{n}}$ for all $i \in \{2,\ldots,n\}$. Note that the query $\mq(a_i, h_j,t_1)$, where $h_j$ is such that $rank(a_i, h_j) = 1$, is not really useful since all the agents will answer ``Yes" to this because of the fact that they have unit-range valuations. Nevertheless, we still use it in \Cref{algo:UB-nonAdap-unitRange} to make it clear that it is very similar to \Cref{algo:UB-nonAdap-Priority}.

\begin{algorithm}[tb]
	{\small \centering
		\noindent\fbox{%
			\begin{varwidth}{\dimexpr\linewidth-4\fboxsep-4\fboxrule\relax}
				\begin{algorithmic}[1]
					\small 
					\Input an instance $\mi =(\mn, \mh, \mpp = (P_1, \ldots, P_n))$ and priorities $\bmp = (p_1, \ldots, p_n)$
					\Output returns a Pareto optimal matching when $p_i=0$ for all $i \in [n]$,  priority-$\bmp$ matching otherwise
					
					\State $\mathbb{G}_{\mi} = (\mn \cup \mh, \me) \gets$ graph induced by $\mi$
					\State $t_1 \gets 1$
					\State $t_i \gets \frac{1}{\sqrt{n}}$, for all $i \in \{2,\ldots,n\}$								
					\For{$e = (a_i, h_r) \in \mn \times \mh$}
					\State $r \gets \text{rank}(a_i, h_j)$
					\If{$\mq(a_i, h_j, t_r)$} \label{algo:unitRange:line:query}
					\State $w_e \gets p_r + t_r$
					\Else 
					\State $w_e \gets p_r$
					\EndIf
					\EndFor	
					\State $\mu \gets$ max-weight matching in $\mathbb{G}_{\mi}$, where weights are $(w_e)_{e \in \me}$
					\label{algo:unitRange:line:muMM}
					\If{$p_i = 0$ for all $i \in [n]$} 
					\State $\mu \gets$ run TTC with $\mu$ (from line~\ref{algo:unitRange:line:muMM}) as the initial endowment		\label{algo:unitRange:line:ttc}		
					\EndIf		
					\State\Return $\mu$    			
				\end{algorithmic}
		\end{varwidth}}
		\caption{An $O(\sqrt{n})$-approximation algorithm for finding the optimal social welfare among Pareto optimal or priority-$p$ matchings for the case when the agents have unit-range valuations.}
		\label{algo:UB-nonAdap-unitRange}
	}
\end{algorithm}

\begin{theorem}\label{thm:UB-nonAdap-unitRange}
	Given an instance $\mi = (\mn, \mh, \mpp = (P_1, \ldots, P_n))$, and a vector of priorities $\bmp = (p_1, \ldots, p_n)$, where $\bmp \in \mathbb{P} \cup \{(0,\ldots,0)\}$, \Cref{algo:UB-nonAdap-unitRange} asks one non-adaptive query per (agent, object) pair and for the case when agents have unit-range valuations returns a
	\begin{enumerate}[label=\roman*),ref=\roman*]
		\item Pareto optimal matching that achieves an $O(\sqrt{n})$-approximation to the optimal welfare among all Pareto optimal matchings when $p_i = 0$ for all $i \in [n]$.
		\item priority-$\bmp$ matching that achieves an $O(\sqrt{n})$-approximation to the optimal welfare among all priority-$\bmp$ matchings when $\bmp \in \mathbb{P}$.
	\end{enumerate}   
\end{theorem}

\begin{proof}
	Just like we did to prove \Cref{thm:UB-nonAdap-PO} let us introduce some notation. Given the priority vector $\bmp$, let $\mu^{*}$ denote the matching that achieves the optimal social welfare among Pareto optimal or priority-$\bmp$ matchings when the agents have unit-range valuations. Let $H_i$ denote the set of agents who are matched to their $i$-th choice in $\mu^{*}$ and have value at least $t_i$ for their partner in $\mu^{*}$. Similarly, let $L_i$ denote the set of agents who are matched to their $i$-th choice in $\mu^{*}$ and have value less than $t_i$ for their partner in $\mu^{*}$.  We define $H = \cup_{i=1}^n H_i$, $L = \cup_{i=1}^{n} L_i$, and, for some $S \subseteq \mn$, $\mu^{*}_S \subseteq \mu^{*}$ be the set of edges $(a_i, h_j)$ such that $a_i \in S$ and $(a_i, h_j) \in \mu^{*}$. 	
	
	Additionally, let  $\mu'$ be the matching that is computed in line~\ref{algo:unitRange:line:muMM} in \Cref{algo:UB-nonAdap-unitRange}. Note that if $\mu$ is the matching returned by \Cref{algo:UB-nonAdap-unitRange}, then $\mu' = \mu$ when $\bmp \in \mathbb{P}$. Also, for the case when $p_i = 0$ for all $i \in [n]$, $\sw(\mu) \geq \sw(\mu')$, since $\mu$ is the matching returned by TTC with initial endowments $\mu'$ and we know that TTC is individually rational \cite{shapley74}. 
	
	Next, it is easy to see that the matching  is Pareto optimal when $p_i = 0$ for all $i \in [n]$ and is a priority-$\bmp$ matching when $\bmp \in \mathbb{P}$. (This can be seen by proceeding exactly like in the second paragraph of the proof of \Cref{thm:UB-Adap}.)
	
	Finally, in order to bound the approximation ratio, we will directly bound the ratio of $\frac{\sw(\mu^*)}{\sw(\mu)}$. To do this, first note that we can proceed exactly like in the proof of \Cref{lemma:highVals-PO} to see that $\sw(\mu^{*}_H) \leq \sqrt{n} \cdot \sw(\mu')$. This is so because the only difference is in the way the $t_i$s are defined for all $i \in [n]$. Second, note that $\sw(\mu^{*}_L) \leq |\mu^{*}_L| \cdot \frac{1}{\sqrt{n}} \leq \sqrt{n}$, where the first inequality follows since $|L_1| = 0$ and every agent in $L_i$ for $i \in \{2, \ldots, n\}$ has value at most $t_i = 1/\sqrt{n}$ for their good, and the second inequality follows since $|\mu^{*}_L| \leq n$. Therefore, we have that,
	\begin{equation}
		\frac{\sw(\mu^*)}{\sw(\mu)} = \frac{\sw(\mu^{*}_H) + \sw(\mu^{*}_L)}{\sw(\mu)}
		\leq \frac{\sqrt{n} \cdot \sw(\mu') + \sqrt{n}}{\sw(\mu)} \leq 2\sqrt{n},
	\end{equation}
	where the last inequality follows since $\sw(\mu) \geq 1$, as there is at least one agent who is matched to a rank-1 edge in $\mu$ (since $\mu$ is either a Pareto optimal or a priority-$\bmp$ matching and hence we can use \Cref{clm:boundB}).	
	%
\end{proof}

\subsubsection{Lower bounds} \label{app:sec:unit-range-LB}

Here we derive a lower bound that is similar to one for unit-sum valuations and show that any deterministic algorithm $\ma$ that produces a Pareto-optimal/rank-maximal/max-cardinality rank-maximal/fair matching and that asks at most one query per (agent, object) pair has $\ml(\ma) \in \Omega(\sqrt{n})$. Note that this implies that \Cref{algo:UB-nonAdap-unitRange} is asymptotically optimal.

\begin{theorem} \label{thm:LB-nonAdap-unitRange}
	Let $X$ denote one of the properties in the set $\{${Pareto-optimal, $\allowbreak$ rank-maximal, max-cardinality rank-maximal, and fair}$\}$. Let $\ma$ be a non-adaptive deterministic algorithm that always produces a matching with property $X$ and asks at most one query per (agent, object) pair. If there are $n$ agents with unit-range valuation functions, then $\ml(\ma) \in \Omega(\sqrt{n})$.
\end{theorem}
\begin{proof}[Proof (sketch)]The proof of this is almost identical to the proof of \Cref{thm:LB-nonAdap}, with the main difference being the way the utility functions $u_0, u_1, \ldots, u_8$ are defined. For this proof, we redefine $u_0, u_1, \ldots, u_8$ the following way, to ensure that they are unit-range valuation functions. 
	\begin{align*}
		&u_0 = (1, c_1^2, 0,\ldots, 0)\\	 
		&u_1 = (1, c_1, 0, \ldots, 0) &&u_2 = (1, 1/2, 0, \ldots, 0) \\
		&u_3 = (1, c_1 + \epsilon, 0, \ldots, 0) &&u_4 = (1, 1/4-\epsilon, 0, \ldots, 0)\\
		&u_5 = (1, c_1^2, 0 \ldots, 0) &&u_6 = (1, c_1 - \epsilon, 0, \ldots, 0)\\ 
		&u_7 = (1, c_1^2, 0, \ldots, 0) &&u_8 = (1, c_1/4, 0, \ldots, 0),
	\end{align*} 	
	Given these definitions, one can now proceed exactly like in the proof of \Cref{thm:LB-nonAdap}. Note that the proof will be simpler in this case because for unit-range valuations, the queries of the form $\mq(\cdot, \cdot, T_{i1})$ are not useful, since by definition all unit-range utility functions have value 1 for the most preferred good. 
\end{proof}

\end{document}